\documentclass[reqno]{amsart}

\pdfoutput=1 

\usepackage[english]{babel}
\usepackage[utf8]{inputenc}
\usepackage{csquotes}
\usepackage[T1]{fontenc}

\usepackage{amsmath}
\usepackage{amssymb}
\usepackage{mathtools}
\usepackage{braket}
\usepackage{amsthm}
\usepackage{MnSymbol}  
\usepackage{amsfonts}
\usepackage{physics}
\usepackage{enumitem}
\usepackage[pdfencoding=auto, psdextra,colorlinks=true]{hyperref}
\hypersetup{urlcolor=blue, citecolor=red, linkcolor=blue}
\usepackage{bookmark}%
\usepackage{todonotes}
\usepackage[capitalise]{cleveref}
\usepackage{comment}
\usepackage{booktabs}
\usepackage[figuresright]{rotating}
\usepackage{tabularx}
\usepackage{multirow}
\usepackage{microtype}

\usepackage{graphicx}
\usepackage{subfig}
\usepackage{wrapfig}
\usepackage{xcolor}
\usepackage{subfig}
\usepackage{tikz}
\usetikzlibrary{calc}
\usetikzlibrary{cd}

\usepackage{etoolbox} %
\AtBeginEnvironment{pmatrix}{\everymath{\displaystyle}}
\AtBeginEnvironment{bmatrix}{\everymath{\displaystyle}}

\DeclareMathOperator{\Pic}{Pic}
\DeclareMathOperator{\Span}{Span}
\DeclareMathOperator{\Bir}{Bir}

\DeclareMathOperator{\Sing}{Sing}

\numberwithin{equation}{section}

\newcommand{\Pj}{\mathbb{P}}
\newcommand{\Z}{\mathbb{Z}}
\newcommand{\N}{\mathbb{N}}
\newcommand{\R}{\mathbb{R}}
\newcommand{\Cp}{\mathbb{C}}

\renewcommand{\vec}[1]{\mathbf{#1}}
\newcommand{\PcrossP}{\Pj^{1}\times\Pj^{1}}
\renewcommand{\epsilon}{\varepsilon}
\renewcommand{\imath}{\mathrm{i}}

\newcommand{\bPhi}{\boldsymbol{\Phi}}

\usetikzlibrary{calc, positioning,shapes,arrows,decorations.markings,patterns,angles,quotes,intersections}
\tikzset{Rightarrow/.style={double equal sign distance,={Implies},-},
	triple/.style={-,preaction={draw, Rightarrow}},
	quadruple/.style={preaction={draw,Rightarrow,shorten >=0pt},shorten >=1pt,-,double,double
		distance=0.2pt}}
\tikzstyle{line} = [draw, -latex']
\tikzset{mynode/.style={draw,circle, inner sep=2pt, outer sep=0pt}}
\tikzset{
wei/.style={circle, minimum size=0.4pt,inner sep=0.8pt},
}
\newcommand\bit{\begin{itemize}}
\newcommand\eit{\end{itemize}}
\newcommand\lan{\langle}
\newcommand\ran{\rangle}
\newcommand{\bp}{\begin{pmatrix}}
\newcommand{\ep}{\end{pmatrix}}
\newcommand\Pa{Painlev\'e }
\newcommand{\mb}{\mathbb}

\newcommand{\oc}[1]{{#1}^{\vee}}
\newcommand{\beq}{\begin{equation}}
\newcommand{\eeq}{\end{equation}}
\newcommand{\al}{\alpha}

\newcommand{\de}{\delta}

\newcommand{\De}{\Delta}
\newcommand{\be}{\beta}
\newcommand{\la}{\lambda}

\newcommand{\Ga}{\Gamma}

\allowdisplaybreaks
\makeatletter 
\def\centerarc[#1](#2)(#3:#4:#5);%
    {
    \draw[#1]([shift=(#3:#5)]#2) arc (#3:#4:#5);
    }

\theoremstyle{plain}
\newtheorem{theorem}{Theorem}[section]
\newtheorem{proposition}[theorem]{Proposition}

\newtheorem{lemma}[theorem]{Lemma}
\theoremstyle{definition}
\newtheorem{definition}[theorem]{Definition}
\newtheorem{notation}[theorem]{Notation}

\theoremstyle{remark}
\newtheorem{remark}[theorem]{Remark}

\newtheorem*{conjecture*}{Conjecture}

\usepackage[backend=biber,sorting=nyt,maxnames=99,giveninits=true,doi=false,url=false,citestyle=numeric-comp]{biblatex}
\title{Determination of the symmetry group for some QRT roots}

\author[G. Gubbiotti]{Giorgio Gubbiotti}
\address[G. Gubbiotti]{Dipartimento di Matematica ``Federigo Enriques'',
	Universit\`a degli Studi di Milano, Via C. Saldini 50, 20133
	Milano, Italy \& INFN Sezione di Milano, Via G. Celoria 16,
	20133 Milano, Italy}
\email{giorgio.gubbiotti@unimi.it}

\author[Y. Shi]{Yang Shi} 
\address[Y. Shi]{College of science and Engineering, Flinders at Tonsley, 
Flinders University, SA 5042, Australia}
\email{yang.shi@flinders.edu.au}


\date{\today}

\begin{document}

\maketitle

\begin{abstract}
    We determine the affine Weyl
    symmetries of some two-dimensional birational maps
    known as QRT roots arising from Kahan--Hirota--Kimura discretisation
    of two different reduced Nahm systems. The main finding is that
    the symmetry types of these discrete systems are subgroups of
    the Weyl groups for Sakai's discrete Painlev\'e
    equations to which the QRT maps are the autonomous limits.
\end{abstract}

\setcounter{tocdepth}{1}
\tableofcontents

\section{Introduction}

In this paper, we consider the geometric and algebraic structure of two
QRT roots arising from Kahan--Hirota--Kimura (KHK) discretisation
\cite{Kahan1993,KahanLi1997,HirotaKimura2000} of reduced Nahm
systems \cite{Hitchinetal1995}. We show that these systems admit,
as symmetry group, certain subgroups of the affine Weyl groups
appearing in Sakai's classification of second-order discrete Painlev\'e
equations \cite{Sakai2001}.
Our results are summarised in the following statements:

\begin{theorem}
    Consider the map:
    \begin{equation}
        \varphi_{1} \colon (x,y)\mapsto\left(y, \frac{x y-1}{2x- y}\right).
        \label{eq:qrt0}
    \end{equation}
   Map $\varphi_{1}$ admits a covariant pencil $p_{1}$
    of biquadratic curves
    whose singular fibres configuration is of type
    $(A_{3} \times A_{2} \times 2 A_{1})^{(1)}$. Moreover,
    to the singular fibre of type $A_{2}^{(1)}$ is associated a discrete
    integrable system with symmetry group of type $A_{2}^{(1)} \subset E_{6}^{(1)}$.
    In particular, we show that $\varphi_{1}$ is an element of quasi-translation
    in the extended affine Weyl group of type $E_{6}^{(1)}$.
    \label{thm:mainA}
\end{theorem}

\begin{theorem}
    Consider the map:
    \begin{equation}
        \varphi_{2} \colon \left( x,y \right)
        \mapsto
        \left(  
            y,\frac{2 x y^2-3 x y+2 x-1}{x y^2-2 y^2+3 y-2}
        \right).
        \label{eq:phi3}
    \end{equation}
    Map $\varphi_{2}$ admits a covariant pencil $p_{2}$ of biquadratic 
    curves whose singular fibres configuration is of type
    $(3A_{3} \times A_{1})^{(1)}$. Moreover, the singular fibre of type 
    $A_{1}^{(1)}$ is associated with a discrete integrable system with a 
    symmetry group of type $B_{3}^{(1)}\subset E_{7}^{(1)}$. In particular, we show that $\varphi_{2}$ is an element of quasi-translation
    in the extended affine Weyl group of type $E_{7}^{(1)}$.
    \label{thm:mainB}
\end{theorem}

\Cref{thm:mainB} will be proved using the following statement
on a map which is a generalisation of the map $\varphi_2$.

\begin{theorem}
    Consider the map:
    \begin{equation}
    \varphi_s \colon (x,y) \mapsto
    \left(
    y,
    {\frac {3sx{y}^{2}-3sxy+2s{y}^{2}-4x{y}^{2}+sx+6yx-4x+2}{%
    2 sxy^{2}-3sy^{2}-2x{y}^{2}+2sx+3sy+4{y}^{2}-s-6y+4}}\right).
    \label{eq:phis}
    \end{equation}
    Map $\varphi_{s}$ admits a covariant pencil $p_{s}$
    of biquadratic curves
    whose singular fibres configuration is of type
    $(2A_{3} \times 2A_{1} \times A_0)^{(1)}$. Moreover to the singular fibre of type $A_{1}^{(1)}$ 
    is associated with a discrete integrable system with a symmetry group of type $B_{3}^{(1)}\subset E_{7}^{(1)}$.
    In particular, we show that $\varphi_{s}$ is an element of quasi-translation
    in the extended affine Weyl group of type $E_{7}^{(1)}$.
    \label{thm:mainC}
\end{theorem}

The two maps $\varphi_1$ and $\varphi_2$ are two QRT roots related to the 
Kahan--Hirota--Kimura (KHK) discretisations~\cite{Kahan1993,HirotaKimura2000,KimuraHirota2000} of 
so-called reduced Nahm systems of tetrahedral and octahedral type 
\cite{Hitchinetal1995}. In particular, \Cref{eq:qrt0,eq:phi3} are 
obtained from the KHK discretisation of the tetrahedral and octahedral 
reduced Nahm systems with the method of 
\cite{VanDerKampCelledoniMcLachlanMcLarenOwrenQuispel2019}.
The interested reader can find more details on this construction, with
the explicit formulas are given in \Cref{app:equivalence}.
Finally, the map $\varphi_s$ is a QRT root 
\cite[Introduction]{Duistermaat2011book} preserving a more general 
pencil than $\varphi_2$. The map $\varphi_s$ will be discussed more 
thoroughly in \Cref{sec:origin}.

In \cite{Duistermaat2011book}, the geometry of QRT maps is discussed in 
the context of Oguiso and Shioda's classification of 74 types of singular 
fibre configurations of rational elliptic surfaces \cite{OS1991}. On the 
other hand, it was shown in \cite{Carsteaetal2017} how different choices 
of singular fibres in a configuration lead to different autonomous \Pa 
equations in Sakai's classification of second-order discrete Painlev\'e
equations \cite{Sakai2001}. 
The approach given in \cite{Carsteaetal2017} allows one to discuss a given 
QRT system within Sakai's framework.

In this paper, we analyse the two discrete systems given by the maps 
$\varphi_1$ and $\varphi_2$ in \Cref{eq:qrt0,eq:phi3} with tools from 
algebraic geometry and theory of the Weyl group. That is, we build the space 
of initial values and obtain the corresponding singular fibre configurations of the two maps. We find that the map 
$\varphi_2$ has a peculiar geometric structure, which is different from 
the canonical one presented in the foundational paper \cite{Sakai2001} 
and in the review \cite{KajiwaraNoumiYamada2017R}.
For some choices of singular fibres from the two configurations, 
we study in detail the symmetries of the associated discrete systems.
In particular, we show that these discrete systems admit symmetries 
that do not appear explicitly in Sakai's classification. Such 
two-dimensional discrete integrable systems have been found in different 
contexts \cite{T:03, KNT:11, CarsteaTakenawa2012, Carsteaetal2017, 
ahjn:16, cp:17, Stokes:18}.  In \cite{Shi:19}, it was found that 
these symmetries arise as normalizers of certain parabolic subgroups of 
Sakai's Weyl groups using the 
normalizer theory of 
of Coxeter groups \cite{H,BH}. 
Sakai's equations are realised as birational 
representations of affine Weyl groups of $ADE$ type in the plane, 
see \cite[Proposition 5.15]{SchuttShioda2019mordell}.  
In particular, the map that gives the
iteration/dynamics of Sakai's equation corresponds to 
translation along a shortest weight on the weight lattice of the corresponding Weyl group \cite{Sakai2001,KajiwaraNoumiYamada2017R}.
Here we show that the QRT roots considered here are 
quasi-translations. That is, an element of an affine Weyl group
of infinite order, which 
after $m$ number of 
iterations become a translation in the original Weyl 
group is called a \emph{quasi-translation of order $m$}.

The plan of the paper is the following: in \Cref{sec:backgrownd} we 
recall some basic facts from the theory of the elimination of the
indeterminacies of maps of surfaces and the theory of Weyl groups. In 
\Cref{sec:origin} we construct the space of initial conditions of the 
maps $\varphi_{1}$ and $\varphi_{2}$ and identify the associated singular 
fibres configurations, that is the first statements in 
\Cref{thm:mainA,thm:mainB}. Then, we proceed to introduce the map 
$\varphi_{s}$ by constructing the QRT root preserving a pencil 
generalising the one preserved by $\varphi_{2}$. We conclude the section 
by computing the configuration of the singular fibres associated to 
$\varphi_{s}$, thus proving the first statement of \Cref{thm:mainC}. Next,
in \Cref{sec:octE6} we prove the second Statement of \Cref{thm:mainA}.
Then, in \Cref{sec:B3} we prove the last statement of \Cref{thm:mainC} 
where we show that to singular fibre of type $A_{1}^{(1)}$ 
associated to the map $\varphi_{s}$~\eqref{eq:phis} corresponds to a
discrete system of type $B_{3}^{(1)}\subset E_{7}^{(1)}$. 
Using the results of \Cref{sec:B3} in \Cref{sec:C3} we are able to
prove the last statement of \Cref{thm:mainB}, i.e.\ that to singular fibre
of type $A_{1}^{(1)}$ associated to the map $\varphi_{2}$~\eqref{eq:phi3} 
corresponds to a discrete system of type $B_{3}^{(1)}\subset E_{7}^{(1)}$.
This result is obtained by showing that the action of the map $\varphi_{2}$
on the root lattice $E_{7}^{(1)}$ is the same as the one of 
$\varphi_{s}$, but on a \emph{translated sub-root lattice}.
Finally, in \Cref{sec:conclusions} we give some conclusions and 
outlook for further work.

\section{Background material}
\label{sec:backgrownd}

In this section, we give the background we need for our study of the
symmetries of the two maps in \Cref{eq:qrt0,eq:phi3}.
In particular, we will introduce the concept we need from the theory
of elimination of the indeterminacy of surfaces following \cite{Beauville1996_2nd}.
Following Sakai \cite{Sakai2001} this allows us to construct the \emph{space of initial values} of a given map, which is the discrete analogue of Okamoto's
description \cite{Okamoto1979,Okamoto1977} of the continuous Painlev\'e
equations \cite{InceBook}. Moreover, we will need some basic facts from
the theory of Weyl groups or more generally, the theory of Coxeter groups \cite{Bbook}.

\subsection{Birational maps and biquadratic pencils}
\label{desing}

In this section, we briefly recall some tools from the algebraic geometry of
birational maps we need to carry out our constructions, namely the elimination 
of indeterminacies for maps of complex surfaces, and some generalities 
on singular fibres of pencils.

\subsubsection{Eliminaton of indeterminacies for maps of surfaces}

As stated in the introduction, we will consider rational maps of $\Cp^{2}$
to itself in the compactification $\PcrossP$. In particular, we will consider
birational maps, i.e. rational maps whose inverse is rational, and denote
their space by $\Bir(\PcrossP)$. An \emph{indeterminacy point} is a point where
a map $\varphi\in\Bir(\PcrossP)$ is not well defined. Representing locally a map as 
pairs of rational functions, being defined corresponds to say that
one of the entries is in the form $0/0$.
Indeterminacies are eliminated through the repeated application of a procedure
called the blow-up at a given point, which  can be described heuristically as 
the variety obtained by replacing a point with the set of lines passing through 
it, see \cite[Section II.1]{Beauville1996_2nd}.
A blow-up is a birational projective morphism, i.e. it is an isomorphism outside 
the given point, and the preimage of that point is called the 
\emph{exceptional curve} (usually denoted by the letter $E$).

Since $\PcrossP$ is a ruled surface  (in particular, the simplest of the Hirzebruch surfaces),
to eliminate the  indeterminacies of $\varphi\in\Bir(\PcrossP)$ we use the following result
from \cite[Corollay II.12]{Beauville1996_2nd} adapted to our case:

\begin{theorem}
    Let $\varphi\in Bir(S)$ with $S$ surface. Then there exists a surface 
    $\widetilde{S}$ and a map $\widetilde{\varphi}\in\Bir(\widetilde{S})$
    such that the following diagram commute:
    \begin{equation}
    \begin{tikzcd}
            \widetilde{S}   \arrow[r, "\widetilde{\varphi}"] 
            & \widetilde{S} \arrow[d, "\epsilon"]
            \\
            S \arrow[u, "\epsilon",leftarrow] \arrow[r, "\varphi",dashed]&  S
        \end{tikzcd}    %
        \label{eq:commblow}
    \end{equation}
    and the morphisms $\epsilon$ is the composition of blow-ups.
    \label{thm:beauville}
\end{theorem}

The map $\widetilde{\varphi}\in\Bir(\widetilde{S})$ in  \Cref{thm:beauville} 
is called the \emph{lift} of $\varphi$ and  if no confusion arise we will 
denote it without the tilde. In the case of $\varphi\in\Bir(\PcrossP)$, following
\Cref{thm:beauville}, we eliminate the indeterminacies blowing up all the points 
of the following classes:
\begin{itemize}
    \item indeterminacy points of the map $\varphi$ and its inverse $\varphi^{-1}$;
    \item indeterminacy points lying on an exceptional line $E$ after blowing
        up, called \emph{infinitely near points}; 
    \item images of lines contracted by $\varphi$ or $\varphi^{-1}$.
\end{itemize}

\begin{remark}
    We remark that while in this paper we are following the
    the classical strategy of elimination of indeterminacies explained
    in \Cref{thm:beauville} this is not the only possible approach.
    Indeed, the blow-up of infinitely near points can
    be achieved with a single blow-up using an approach based on the
    ideals of these points, see \cite{GraffeoRicolfi2023}.
    \label{rem:blowups}
\end{remark}

So, to eliminate the indeterminacies  of a map $\varphi\in\Bir(\PcrossP)$
we have to blow-up a set of the form\footnote{This is a slight abuse of notation
because after blowing up points are replaced by lines, but we use them to
keep track of which points originated a ``chain'' of successive blow-ups.}:
\begin{equation}
    \Sing (\varphi) = 
    \left\{ b_{1,\dots,i_{1}},b_{i_{1}+1,\dots,i_{2}},\dots,b_{i_{L-1}+1,\dots,i_L }\right\},
    \label{eq:singgen}
\end{equation}
where $i_1<i_2<\ldots<i_{L-1}<i_L=M$, $L,M$ are positive integers,
and we denoted some infinitely near points by 
$b_{i_{k},\ldots,i_{k+1}-1}$.
Then on the space obtained by blowing up the set $\Sing (\varphi)$,
written as $B_{\varphi}$, the map $\widetilde{\varphi}$ is an automorphism. 
The smooth variety $B_{\varphi}$ is called the \emph{space of initial 
values of the map $\varphi$} \cite{Sakai2001,Takenawa2001JPhyA}.
This definition is stated in analogy with the same definition for
continuous systems given by Okamoto \cite{Okamoto1977,Okamoto1979}.

\subsubsection{The Picard group}

Given a smooth (quasi-)projective variety $X$, recall that 
\cite[Fact I.1]{Beauville1996_2nd} 
the Picard group of $X$, $\Pic(X)$, is the group of isomorphism classes
of line bundles on $X$. Line bundles can be identified with divisors up 
to linear equivalence, and throughout the paper we will make use of 
this identification. In the case of a surface $B$ obtained blowing up 
$\PcrossP$ $M$ times the Picard group is the free $\Z$-module 
generated by the proper transform of classes 
of equivalence of horizontal and vertical lines:
\begin{equation}
    H_{x} = \epsilon^{*}\left\{ x=\text{const} \right\},
    \quad
    H_{y} = \epsilon^{*}\left\{ y=\text{const} \right\},
    \label{eq:HxHy}
\end{equation}
and the exceptional lines $E_{1},E_{2},\dots,E_{M}$:
\begin{equation}
    \Pic\left( B \right)=
    \Span_{\Z}( H_{x}, H_{y}, E_{1},\dots,E_{M}).
    \label{eq:picardgen}
\end{equation}

\begin{notation}
    We will denote elements of the Picard group by capital letters while 
    their (possible) representatives (complete linear systems of) curves will be 
    denoted by small letters, e.g. $D\in\Pic(B)$ and $d\subset \PcrossP$
    is a curve.
\end{notation}

The Picard group is equipped with a symmetric bilinear form $\left( \,,\, \right)$
with integer values called the \emph{intersection form}, see \cite[Definition I.3 
and Theorem I.4]{Beauville1996_2nd}. Let $B$ be a surface as above, and then the 
intersection form satisfies the following rules:
\begin{equation}
    \begin{gathered}
        \left( H_{x},H_{x} \right)=\left( H_{y},H_{y} \right)=
        \left( H_{x},E_{i} \right)=\left( H_{y},E_{i} \right)=0,
        \\
        \left( H_{x},H_{y} \right) = 1,
        \quad
        \left( E_{i},E_{j} \right) = -\delta_{i,j},
    \end{gathered}
    \label{eq:intform}
\end{equation}
where $i,j=1,\ldots,M$, see \Cref{eq:picardgen}.

We will denote the action induced by a morphism $\varphi \colon B_\varphi\to B_\varphi$
on the Picard group by $\varphi^{*} \colon \Pic(B_\varphi)\to \Pic(B_\varphi)$. This action
is \emph{linear} and can be represented by a matrix with integer entries.

\subsubsection{Rational elliptic fibrations and minimality}

It is known that integrable maps in the plane are related 
to \emph{rational elliptic surfaces} or their generalisation 
\cite{CarsteaTakenawa2012,Tsuda2004,Bellon1999}.
We recall that a surface is elliptic if it admits an elliptic fibration.
Such a surface is a rational elliptic surface if the fibres are connected,
see \cite[Theorem 9.1.3 and Definition 9.1.4]{Duistermaat2011book}.
This implies that there these maps are birationally equivalent
to a map $\varphi_\text{min}$ such that $\abs{\Sing (\varphi_\text{min})}=8$, 
that is $i_L=8$ in equation \eqref{eq:singgen},
see also \cite{CarsteaTakenawa2013}.
For a ``generic'' map whose space of initial values is
a rational elliptic surface
we have $\abs{\Sing (\varphi)}\geq 8$. From now on, we will assume that we are
considering the ``minimal'' map such that this lower bound is obtained,
i.e. $\varphi = \varphi_{\text{min}}$.
In the case of a rational elliptic surface $B_\varphi$ the
\emph{anti-canonical divisor} takes the following form:
\begin{equation}
    -K_{B_\varphi} = 2 H_x + 2 H_y - \sum_{k=1}^{8} E_i.
    \label{eq:antican}
\end{equation}
Moreover, in this case the anti-canonical divisor, or a multiple
of it, is represented by a \emph{pencil of biquadratic curves} whose base points are the points $b_k$ with $k=1,\ldots,8$.
We recall that a point $b\in\PcrossP$ is a \emph{base point} for a pencil $p$
if $p(b;\mu:\nu)=0$ regardless of the value of $[\mu:\nu]\in\Pj^1$.
The map on the Picard group $\varphi^{*}$ preserves the
anti-canonical divisor, that is:
\begin{equation}
    \varphi^{*}(-K_{B_\varphi}) = -K_{B_\varphi}.
\end{equation}
This, in turn implies that taking the ratio of the biquadratic polynomials 
defining the pencil $p$, i.e. the rational function $\mathcal{H}=q/r$, we 
obtain an invariant for the map $\varphi\in\Bir(\PcrossP)$, or
a $k$-invariant in the case the map does not preserve the fibres of
the pencil \cite{CarsteaTakenawa2012}. In the rest of the paper
we will only consider fibre-preserving actions, we will not need the notion of $k$-invariant we will not enter into the
details of this approach, but we refer the interested reader to the
papers \cite{Haggaretal1996,RJ15,HietarintaBook,GG_cremona3}. So, from
now on when we will say \emph{covariant pencil $p$} we will mean that 
$\varphi^{*}(p) = \kappa p$ for some $\kappa$ bihomogeneous polynomial
and additionally the \emph{fibres of the pencil are preserved}.

\begin{remark}
    An alternative construction of integrable maps
    starts from a given pencil of biquadratic maps and then considers
    the blow-up of its base points. This yields the
    rational elliptic surface $B_\varphi$ again. The map $\varphi\in\Bir(\PcrossP)$ is
    obtained from the pencil through the application of some properly defined
    involutions, see \cite{QRT1988,QRT1989,Duistermaat2011book}. 
    \label{rem:blowuppencil}
\end{remark}

\subsubsection{Singular fibres and root systems}

Before going on we recall the following definition of singular fibres of
a pencil:

\begin{definition}
    Let 
    \begin{equation}
        p(x,y;\mu:\nu) = \mu p_{0}(x,y) + \nu p_{1}(x,y),
        \label{eq:pencil}
    \end{equation}
    be a \emph{pencil} of plane curves.
    A point $\left(x_{0},y_{0};\mu':\nu' \right)\in\PcrossP\times\Pj^{1}$ 
    such that
    \begin{equation}
        p\left( x_{0},y_{0};\mu':\nu' \right) = 
        \frac{\partial p}{\partial x}\left( x_{0},y_{0};\mu':\nu' \right)=
        \frac{\partial p}{\partial y}\left( x_{0},y_{0};\mu':\nu' \right)=0,
        \label{eq:singdef}
    \end{equation}
    where the derivatives are taken in a properly chosen chart,
    is called a \emph{singular point} for the pencil \eqref{eq:pencil}.
    Given a pencil of plane curves $p$
    if the curve $p\left( x,y;\mu':\nu' \right)$,
    with fixed $[\mu':\nu' ]\in\Pj^{1}$ contains a singular
    point then it is called a \emph{singular fibre} of the pencil
    \eqref{eq:pencil}.
    \label{def:singular}
\end{definition}

The singular fibres of rational elliptic surfaces have been extensively
studied in the literature, see \cite{SchuttShioda2019mordell} for a general
overview and \cite[Chap. 7]{Duistermaat2011book} for some applications to integrable
systems. Here, following \cite{KajiwaraNoumiYamada2017R}, we limit ourselves to note 
that to a singular fibres of a pencil $p$ on the space $B_\varphi$ the anti-canonical 
divisor splits into several \emph{rational} components. That is, given a singular fibre $[\mu':\nu']\in\Pj^1$
there exists some divisors $D_j$, $j=0,\ldots,\kappa$ such that:
\begin{equation}
    -K_{B_\varphi} = \sum_{j=0}^{\kappa} D_j,
\end{equation}
which realise this singular fibre.
Moreover, the intersection matrix:
\begin{equation}
    C = \left(-(D_i,D_j)\right)_{i,j=0}^{\kappa},
\end{equation}
is a generalised Cartan matrix of a Coxeter group of rank $\kappa$, see \cite[Proposition 5.15]{SchuttShioda2019mordell}. We call the type of root system defined by this
generalised Cartan matrix \emph{the surface type},
denoted by $\Lambda$, for a singular
fibre for the map $\varphi$. 
That is, the divisors $D_k$ correspond to the 
simple roots of $\Lambda$.
The orthogonal complement of $\Lambda$, $\Lambda^\perp$,
is another root system, called the \emph{symmetry type}
of the map $\varphi$.
Note that the surface type associated with the map $\varphi$ depends on the 
choice of the singular fibre as noted in \cite{CarsteaTakenawa2013}.
Indeed, the possible surface type associated with the singular fibres are those
contained in \cite[Theorems 8.8 and 8.9]{SchuttShioda2019mordell}.

To summarise, for a given birational map
$\varphi\in\Bir(\PcrossP)$, one can eliminate its indeterminacies
and those of its inverse $\varphi^{-1}$ and build
the space of initial conditions $B_\varphi$. If the map is integrable,
then $B_\varphi$ is a rational elliptic surface, and to each of its
singular fibres there are associated with two orthogonal root systems
$\Lambda$ and $\Lambda^\perp$. The simple roots of these systems give the surface
and the symmetry type, respectively, and can both be represented by
Dynkin diagrams.  The pullback $\varphi^{*}$ acts on the simple roots
of $\Lambda$ as a permutation, while in general, it acts on the simple
roots of $\Lambda^\perp$ as a quasi-translation. In
the next subsection, we will recall some facts from the theory of Weyl
groups needed to discuss the natures of the maps $\varphi_1$ and $\varphi_2$ and the symmetries of the associated discrete systems.

\subsection{Weyl group}

The symmetry groups of affine Weyl type for the two discrete systems given by maps
\eqref{eq:qrt0} and \eqref{eq:phi3} can be constructed using the
geometric characterisation given in Section \ref{desing}. Here we 
recall necessary facts and properties of the Weyl group in order for us
to construct the groups of symmetries and the
corresponding group elements that are associated to
the two maps \eqref{eq:qrt0} and \eqref{eq:phi3}.

\subsubsection{Dynkin diagram and root system}
\label{Dyn}
Let $\Gamma$ be
a Dynkin diagram with $n$ nodes, its associated \emph{Coxeter group}
is a generating set,
\begin{equation}
    W=W(\Ga)= \langle s_i\mid\; 1\leq i \leq n\rangle
\end{equation}
with \emph{defining relations}:
\begin{equation}\label{funW}
s_i^2=1,\quad (s_is_j)^{m_{ij}}=1, \; 1\leq i<j \leq n. 
\end{equation}
When $m_{ij}\in\{2,3,4,6\}$, known as the 
\emph{crystallographic condition}, $W$ is called a \emph{Weyl group}. 
Each node of $\Ga$, labelled $i$ represents the generator $s_i$ of order $2$. 
The parameter $m_{ij}$ takes value of: $2, 3, 4$ or $6$ 
when two nodes labeled $i$ and $j$
are: disconnected, joined by a single, a double, or a triple bond, respectively. Diagrams which have only
single bonds are called {\it simply-laced}, they are of types $A_n$, $D_n$, $E_6$, $E_7$ and $E_8$. 
The non-simply laced types are $B_n$, $C_n$, $F_4$ and $G_2$ 
(see Humphries \cite[Sec. 2.5]{Hbook} for example).
Let $\Ga^{(1)}$ be the \emph{affine extension of $\Ga$} (with
an extra node labelled $0$), then
$W^{(1)}=\langle s_i\mid\; 0\leq i \leq n\rangle$ is the associated \emph{affine Weyl group}. 

Let $V^{(1)}$ be an $n+1$-dimensional real vector space
with a basis, 
\beq
\De^{(1)}=\De(\Ga^{(1)})
=\{ \al_i\in V^{(1)}\mid 0\leq i \leq n\},
\eeq
known as
the {\it simple system}, and $\al_i$ ($0\leq i \leq n$) are the {\it simple roots} of $W^{(1)}$. Note that
$\De=\De(\Ga)
=\{ \al_i\in V^{(1)}\mid 1\leq i \leq n\}$
is the simple system of $W$, and
$V=\Span(\De)$ is a subspace of $V^{(1)}$.

A semidefinite symmetric bilinear form 
on $V^{(1)}$ can be defined using the \emph{Cartan matrix} of
type $\Ga^{(1)}$,
$C(\Ga^{(1)})=\left(a_{ij}\right)_{1\leq i,j\leq n, 0}$ by,
\begin{equation}\label{alaij0}
\al_i\cdot\al_j=\frac{|\al_j|^2a_{ij}}{2},
\quad\mbox{for}\quad i,j\in\{0,1,..., n\},
\end{equation}
where the quantity $\al_j.\al_j=|\al_j|^2$ gives the usual interpretation of
 squared length of $\al_j$. The generator $s_j\in W^{(1)}$
known as a {\it simple reflection}, can
be realised as the reflection 
along $\al_j\in V^{(1)}$, its action on $\De^{(1)}$
given by 
\begin{equation}\label{sij0}
s_j(\al_i)=s_{\al_j}(\al_i)
=\al_i-\frac{2(\al_i\cdot\al_j)}{\al_j\cdot\al_j}\al_j
=\al_i-a_{ij}\al_j,\quad \mbox{for all}\quad i,j\in\{0,1,..., n\},
\end{equation}
where $a_{ij}$ is the $(i,j)$-entry
of $C(\Ga^{(1)})$.

Note that nodes of a Dynkin diagram
can be equivalently associated with either simple reflections or simple roots of the Weyl group.

There is a {\it null root} $\de$ in $V^{(1)}$
such that
\beq\label{aid}
\al_i\cdot\de=0\quad\mbox{for all}\quad \al_i\in V^{(1)},\quad 0\leq i \leq n.
\eeq
In particular, we have
\beq\label{d}
\de=\al_0+\tilde{\al}=\al_0+\sum_{i=1}^{n}c_i\al_i=\sum_{i=0}^{n}c_i\al_i,
\eeq
where $\tilde{\al}$ is the {\it highest long root} of $W$,  and values of 
the $c_i$'s can be found in any
classic textbooks on Coxeter groups such as Bourbaki \cite[Planche I--X]{Bbook}.
Finally, the root system of $W^{(1)}$ is given by
\beq\label{arsa}
\Phi^{(1)}=\Phi(\Ga^{(1)})=W^{(1)}\cdot\De^{(1)}
=\{\al+m\de\mid\al\in\Phi,\; m\in\mathbb{Z}\},
\eeq
where 
$\Phi=W.\De$ is the finite root system of
$W$ in $V$. Moreover, we have $\Phi^{(1)}=\Phi^{(1)}_{+}\cup\Phi^{(1)}_{-}$ where,
\beq\label{prs2a}
\Phi^{(1)}_{+}=\{\al=\sum_{i=0}^n \la_i\al_i\mid\mathbb{Z}\ni\la_i\geq 0, \al_i\in \De^{(1)}\}\quad\text{and}
\quad\Phi^{(1)}_{-}=\{-\al\mid\al\in \Phi^{(1)}_{+}\},
\eeq
and $\Phi_{+}$ is the {\it positive root system}.

In general, the element of reflection along any root 
$\be\in\Phi^{(1)}$, denoted by $s_{\be}\in W^{(1)}$,
acts on $V^{(1)}$ by,
\beq\label{sbv}
s_{\be}(v)=v-\frac{2v\cdot\be}{\be\cdot\be} \be,
\eeq
for all $v\in V^{(1)}$. Furthermore, $s_\be$ is related to a simple reflection $s_i$
by the formula
\beq\label{sr}
s_{\be}=s_{w(\al_i)}=ws_iw^{-1},\quad\mbox{for some}\quad w(\al_i)=\be,\quad w\in W^{(1)}.
\eeq

\subsubsection{Dual space, coroots and weights}\label{dual}
By Equation \eqref{d}, we see that 
$\{\al_1, \al_2, \ldots, \al_n, \de\}$
forms another basis of $V^{(1)}$.
Let $V^{(1)*}$ be an $(n+1)$-dimensional real vector space, and $\langle
{}\,, {} \rangle:$ $V^{(1)}\times V^{(1)\ast} \to \mathbb R$ be a bilinear
pairing between $V^{(1)}$ and $V^{(1)\ast}$.  Let $\{ h_1, \ldots, h_n,
h_\de\} $  be  a basis of $ V^{(1)\ast}$   dual to $\{\al_1, \al_2,
\ldots, \al_n, \de\}\subset V^{(1)}$. This is,
\begin{subequations}
    \begin{align}
        &\langle \al_i, h_j\rangle=\delta_{ij},\\
        &\langle \al_i, h_{\de}\rangle=
        \langle \de, h_{j}\rangle=0,\quad \mbox{for}\quad 1\leq i, j\leq n,
        \\
        &\langle \de,h_\de \rangle=1.
    \end{align}
    \label{ah2}
\end{subequations}

The group $W^{(1)}$ acts on $V^{(1)\ast}$ via the contragredient action:
\beq\label{cona}
\lan w^{-1} v,  h\ran=\lan v, wh\ran, \quad\mbox{for}\quad v\in V^{(1)}, h\in V^{(1)\ast}, w\in W^{(1)}.
\eeq

    The vectors $h_j\in V^{(1)\ast}$ $( 1\leq j\leq n)$ are called 
    \emph{the fundamental weights} of $W^{(1)}$,
    and the \emph{weight lattice} is given by,
    \begin{equation}\label{wl}
    P=\Span_{\Z}(h_1,\ldots,h_n).
    \end{equation}

$W^{(1)}$ acts on $h_j$ $( 1\leq j\leq n)$ by,
\begin{subequations}
  \beq\label{sihjn}
s_i(h_j)=
 \begin{cases}
 h_j, &\text{for}\quad i\neq j,\\
 h_j-\sum\limits_{k=1}^n a_{ki}h_k, &\text{for}\quad i=j,\quad\mbox{for}\quad
 1\leq i,j\leq n\\
 \end{cases}
 \eeq
 and
 \beq\label{s0hj}
s_0(h_j)=h_j+c_j\sum\limits_{k=1}^n a_{k0}h_k,
 \eeq  
\end{subequations}
where $c_i$'s are the coefficients of $\al_i$ in $\de$, and $a_{ij}$ are the entries of $C(\Ga^{(1)})$.

In \cite{Shi:22}, we discussed some useful properties and formulas 
of the Weyl groups that enable us to describe quantitatively 
different types of translational elements of the Weyl group. Here, we simply state the
results needed from \cite{Shi:22}.
First we define an $n$-dimensional hyperplane in $V^{(1)*}$,
\begin{equation}\label{Xk}
X_k=\{\,h\in V^{(1)\ast}\,|\, \langle \de, h\rangle=k, k\in \mathbb{R}\}.
\end{equation}

Note that, $X_0$ is an $n$-dimensional vector space, and by \Cref{sihjn} we see that $\{h_j\mid 1\leq j\leq n\}$ forms a basis of $X_0$. Moreover, we have $X_1=h_\de+X_0$ is the affine plane
on which translations of the affine Weyl group
can be realised.

For each $\be\in\Phi^{(1)}$, let
\beq\label{ocal}
\oc\be=\frac{2\be}{\be\cdot\be}=\frac{2\be}{|\be|^2}.
\eeq 
The set of {\it simple coroots}
$\{\pi(\oc\al_j)\in V^{(1)*}\mid 1\leq j\leq n\}$  is a basis of $X_0$ given by,
\beq\label{pah}
\pi(\oc\al_j)=\sum_{k=1}^n a_{kj}h_k=\sum_{k=1}^n\left(C(\Ga)^T\right)_{jk}h_k,\quad
\mbox{for}\quad1\leq j\leq n,
\eeq
or we have
\beq\label{hpa}
h_i=\sum_{k=1}^n\left(C(\Ga)^T\right)^{-1}_{ik}\pi(\oc\al_k),\quad
\mbox{for}\quad1\leq i\leq n.
\eeq

The following proposition enables us to discuss the lengths of vectors
in the subspace $X_0\subset V^{(1)*}$. 

\begin{proposition}
    \label{sbfX0}
    The bilinear pairing $\langle\,,\,\rangle$ given in \Cref{ah2}
    can be restricted to a symmetric and positive definite bilinear form
    $({}\,,{})$: $X_0\times X_0\to \mathbb{R}$, $X_0\subset V^{(1)*}$,
    satisfying:
    \beq\label{paipj}
    \left(\pi(\oc\al_i),\pi(\oc\al_j)\right)
    =\frac{2}{|\al_i|^2}a_{ij}, \quad 1\leq i, j\leq n.
    \eeq
\end{proposition}

In the $\{h_j\mid1\leq j\leq n\}$ basis we have,
\beq\label{hij}
(h_i, h_j)=\sum_{k=1}^n\left(C(\Ga)^T\right)^{-1}_{ik}\frac{2}{|\al_k|^2}\de_{kj},\quad
1\leq i, j\leq n.
\eeq
\begin{remark}\label{iden}
    One can identify the simple coroots with simple roots by 
    $\pi(\oc\al_i)=2\al_i/|\al_i|^2$. Then we see that for simply-laced type systems,
    we have $\pi(\oc\al_j)=\al_j$ ($1\leq j\leq n$), and 
    $\pi(\oc\al_0)=-\tilde{\al}$.
\end{remark}
Finally, let $l(w)$ be the \emph{length} of 
$w\in W^{(1)}$, we have,
\beq\label{lfun}
l(ws_{\al})=
 \begin{cases}
 l(w)+1, &\text{if}\quad w(\al)\in\Phi^{(1)}_{+}\\
 l(w)-1, &\text{if}\quad w(\al)\in\Phi^{(1)}_{-},
 \end{cases}
 \eeq 
where $\al\in\De^{(1)}$.
Equation \eqref{lfun}
can be applied repeatedly until we can write $w$ as a product of $k$ simple reflections,
\begin{align}
ws_{l_1}...s_{l_k}&=1,\quad l_1, ..., l_k\in\{0,1,...,n\},\\\nonumber
\mbox{that is,}\quad\quad    w&=s_{l_k}...s_{l_1}.
\end{align}
\subsubsection{Translations}\label{tran}
It is well-known that $W^{(1)}$ suitably extended
contains an
abelian group of translations
on the weight lattice (which is isomorphic to the weight lattice
$P$), 
\beq\label{PtoUd}
\widetilde{W}^{(1)}=A\ltimes{W}^{(1)}=
W\ltimes P=W\ltimes\lan u_j\, \mid 1\leq j\leq n\ran,
\eeq
where $A$ is a group of diagram automorphsims of
the affine Dynkin diagram ${\Ga}^{(1)}$, and
$u_j\in \widetilde{W}^{(1)}$ is defined to be the element that acts on $X_1$ 
by a translation of the fundamental weight $h_j$,
\begin{equation}
    u_j(X_1)=X_1+h_j, \quad h_j\in X_0, 
    \quad 1\leq j\leq n.
\end{equation}
In general, actions of $u_j$ on $V^{(1)*}$ and $V^{(1)}$
are given by the following proposition.

\begin{proposition}
\label{ujf}
The element $u_j\in \widetilde{W}^{(1)}$ acts on $V^{(1)*}$ by
\begin{subequations}
\beq\label{ujhi}
u_j(h_i)=h_i,\quad \mbox{and}\quad u_j(h_\de)=h_\de+h_j,\quad
\quad 1\leq i, j \leq n,
\eeq
\end{subequations}
while on $V^{(1)}$ we have for $1\leq i, j \leq n$,
\begin{subequations}\label{uja}
\beq\label{ujai}
u_j(\al_i)=
 \begin{cases}
 \al_i, &\text{for}\quad i\neq j,\\
 \al_j-\de, &\text{for}\quad i=j,\\
 \end{cases}
 \eeq
 and
 \beq\label{uja0}
u_j(\al_0)=\al_0+c_j\de,
 \eeq
\end{subequations}
where $c_j$ is the coefficient of $\al_j$ in $\de$.
\end{proposition}
In particular, if we have $s(h_j)=\sum_{i=1}^n n_{ij}h_i$ for $s\in {W}$ and some $n_{ij}\in\mathbb{Z}$ then,
\begin{equation}\label{ujnj}
su_js^{-1}=\prod_{i=1}^n u_i^{n_{ij}}.
\end{equation}

Finally, 
any element $u$ of $\widetilde{W}^{(1)}$ can be written in the form $aw$, with
$a\in A$, and $w\in {W}^{(1)}$ with $l(w)=k$.
That is, we have
\begin{equation}\label{uaw}
u=aw=as_{l_k}...s_{l_1},\quad l_1, ..., l_k\in\{0,1,...,n\},
\end{equation}
where $s_j$ are the simple reflections of ${W}^{(1)}$.

\begin{remark}\label{not}
    We will use the notation $\al_i+...+\al_j=\al_{i...j}$ 
    for a sum of simple roots, and $s_i...s_j=s_{i...j}$ 
    for a product of simple reflections, from time to time for simplicity.
For example, the null root of $E_6^{(1)}$ (see Equation
\eqref{deE6}) can be written as $\de=\al_{015224466333}$, while a product of simple reflections  such as $s_{4}s_5s_3s_4$
is now $s_{4534}$.
\end{remark}
\subsubsection{Normalizers
and quasi-translations}\label{norm}
For $\widetilde{W}^{(1)}$ with a simple system $\De^{(1)}$,
let $J\subset \De^{(1)}$,
the corresponding reflection group $W_J=\langle s_i\mid \al_i\in J\rangle$ 
is called the {\it standard parabolic subgroup} of $\widetilde{W}^{(1)}$. The {\it normaliser} of $W_J$ in $\widetilde{W}^{(1)}$ is defined by 
\begin{equation}
  N(W_J)=\{g\in \widetilde{W}^{(1)}\mid g^{-1}W_Jg=W_J\}.  
\end{equation}
In \cite{H, BH} it is shown that 
\begin{equation}
   N(W_J)=N_J\ltimes W_J,\quad
   \mbox{where}\quad N_J=\{w\in \widetilde{W}^{(1)}\mid wJ=J\}.
\end{equation}
In particular, the group $N_J$ is generated by 
{\it R-elements} and {\it M-elements}.
The R-elements, also known as {\it quasi-reflections}, are in general involutions 
 that act permutatively on the subset $J$, whereas 
the M-elements permute the R-elements.
Together, they generate a group of extended affine Weyl type
for which translations (or quasi-translations) can be defined.
In Sections \ref{Ne6} and \ref{Ne7s}
we construct the corresponding $N(W_J)$'s
of which the QRT roots are elements.

Let $V_J=\Span\left(J\right)$, and $V_J^\perp$ be the orthogonal
complement of $V_J$ in $V^{(1)}$, that is, $V^{(1)}=V_J\bigoplus
V_J^\perp$. Then on $V_J^\perp$, the R-elements and quasi-translations behave as reflections and translations, respectively. In Sections \ref{Te6} and \ref{Te7},
we find a basis for $V_J^\perp$,
which forms a root system for $N_J$ and
a subroot system of $\Lambda^\perp$.
Using properties of this subroot system we show that QRT roots are elements of quasi-translations in $N(W_J)$.

\section{The space of initial conditions of the maps $\varphi_{1}$, $\varphi_{2}$, and $\varphi_{s}$}
\label{sec:origin}

In this section, we build the space of initial conditions of the maps
$\varphi_{1}$~\eqref{eq:qrt0} and $\varphi_{2}$~\eqref{eq:phi3}, and
introduce the map $\varphi_{s}$~\eqref{eq:phis} as a generalisation 
of $\varphi_{2}$. Additionally, we identify the configuration of the
singular fibres of the pencils associated to these three maps 
within Oguiso and Shioda's classification of 74 types of singular 
fibre configurations of rational elliptic surfaces \cite{OS1991}. We
will mainly follow the exposition of this theory from the book
\cite{SchuttShioda2019mordell}.

\subsection{The space of initial values of $\varphi_{1}$}
\label{Nahm112}

In this section, we prove the first statement of \Cref{thm:mainA} on the
existence of a covariant pencil with singular fibres of type $(A_{3}\times
A_{2} \times 2A_{1})^{(1)}$ for the map $\varphi_{1}$ \eqref{eq:qrt0}.
We do so by constructing a space of initial values for such a map.
Following \Cref{sec:backgrownd}, a space of initial values of the map
$\varphi_{1}$ \eqref{eq:qrt0} is constructed blowing up its singular
locus. By direct computations, we have:
\begin{equation}
    S_{1}:=\Sing \varphi_{1} = \left( b_{1},\dots,b_{8} \right)
    \label{eq:sing1}
\end{equation}
where:
\begin{equation}
    \begin{gathered}
        b_{1}\colon (x,y) =\left( \sqrt{2}, \frac{1}{\sqrt{2}}\right),
        \quad
        b_{2}\colon (x,y) =\left( -\sqrt{2}, -\frac{1}{\sqrt{2}} \right),
        \\
        b_{3}\colon (x,y) =\left( -\frac{1}{\sqrt{2}}, -\sqrt{2} \right),
        \quad
        b_{4}\colon (x,y) =\left( \frac{1}{\sqrt{2}}, \sqrt{2}\right),
        \\
        b_{5}\colon (x,y) =\left(\frac{1}{\sqrt{2}},0  \right),
        \quad
        b_{6}\colon (x,y) =\left(-\frac{1}{\sqrt{2}},0  \right),
        \\
        b_{7}\colon (x,y) =\left(0,-\frac{1}{\sqrt{2}}  \right),
        \quad
        b_{8}\colon (x,y) =\left(0,\frac{1}{\sqrt{2}}  \right).
    \end{gathered}
    \label{eq:bpqrt0}
\end{equation}

We remark that not all the base points are immediately visible as
indeterminacy points of $\varphi_{1}$ and its inverse, but some appear
as a contraction of the exceptional lines associated with a point after
blowing up. We call $E_{i}$ the exceptional line associated
to the blowing up of the indeterminacy point $b_i$. Then we have that
$b_{i}$, $i=5,6,7,8$, appear as image of the exceptional lines $E_{1}$, $E_{2}$ and $E_{3}$,
$E_{4}$ under $\varphi_{1}$ and $\varphi_{1}^{-1}$ respectively.

Then, a space of initial values is $X_{1} = B_{\varphi_{1}}$,
a rational elliptic surface. A biquadratic pencil 
is defined by the condition of passing through the points $b_{i}$:
\begin{equation}
    p_{1}(x,y;\mu:\nu) = \mu[\sqrt{2} (x-y)+1][ \sqrt{2}( x-y)-1] + \nu x y (x y-1).
    \label{eq:p1}
\end{equation}
By construction this pencil is covariant with respect to the map
$\varphi_{1}$ \eqref{eq:qrt0}.

Then, from \Cref{def:singular} we have that there are four distinct
singular fibres corresponding to the values of $\left[ \mu:\nu
\right]$ shown in the first column of \Cref{tab:singoct}. The irreducible
components of the associated singular fibres are shown in the second
column of the same table.

\begin{table}[hbt]
    \centering
    \begin{tabular}{ccccc}
        \toprule
        Point & 
        Irred. comp. & 
        Repr. in $\Pic(X_{1})$ &
        Surf./symm. type
        \\
        $[\mu:\nu]$ & $d_{i} $ & $D_{i}$ & $\Lambda$/$\Lambda^\perp$
        \\
        \midrule
        \multirow{3}*{$\left[0:1\right]$} & 
        $\left\{ x =0  \right\}$ & $H_{y}-E_{5}-E_{6}$ &
        \multirow{3}*{$A_{2}^{(1)}$/$E_{6}^{(1)}$}
        \\
        & $\left\{ y =0  \right\}$ &$H_{x}-E_{7}-E_{8}$,
        \\
        & $\left\{ xy = 1  \right\}$ & $H_{x}+H_{y}-E_{1}-E_{2}-E_{3}-E_{4}$
        \\
        \midrule
        \multirow{2}*{$\left[1:0\right]$} & 
        $\left\{y -x =-1/\sqrt{2}  \right\}$ & $H_{x}+H_{y}-E_{1}-E_{3}-E_{5}-E_{7}$
        & \multirow{2}*{$A_{1}^{(1)}$/$E_7^{(1)}$}
        \\
        & 
        $\left\{y -x =1/\sqrt{2}  \right\}$ & $H_{x}+H_{y}-E_{2}-E_{4}-E_{6}-E_{8}$
        \\
        \midrule
        \multirow{4}*{$\left[1:4\right]$} 
        & $\left\{ x = -1/\sqrt{2} \right\}$ & $H_{x}-E_{3}-E_{6}$ 
        & \multirow{4}*{$A_{3}^{(1)}$/$D_5^{(1)}$}
        \\
        & $\left\{ y=1/\sqrt{2} \right\}$ & $H_{y}-E_{1}-E_{8}$ 
        \\
        & $\left\{x =1/\sqrt{2} \right\}$ & $H_{x}-E_{4}-E_{5}$ 
        \\
        & $\left\{ y= -1/\sqrt{2} \right\}$ & $H_{y}-E_{2}-E_{7}$ 
        \\
        \midrule
        \multirow{2}*{$\left[1:16\right]$} 
        & $\left\{4 x y+ \sqrt{2}\left(x+y\right) =1 \right\}$
        & $H_{x}+H_{y}-E_{2}-E_{3}-E_{5}-E_{8}$
        & \multirow{2}*{$A_{1}^{(1)}$/$E_7^{(1)}$}
        \\
        & $ \left\{4 x y-\sqrt{2}\left(x+y\right)=1\right\}$
        & $H_{x}+H_{y}-E_{1}-E_{4}-E_{6}-E_{7}$
        \\
        \bottomrule
    \end{tabular}
    \caption{Irreducible components of the singular fibres in the pencil 
    associated to \eqref{eq:qrt0H}, their representative in the Picard groups,
    and the associated surface and symmetry type.}
    \label{tab:singoct}
\end{table}

\begin{remark}
    From the covariant pencil \eqref{eq:p1} we can build the rational 
    function:
    \begin{equation}
        \mathcal{H}_{1}
        =
        \frac{[\sqrt{2} (x-y)+1][ \sqrt{2}( x-y)-1]}{x y (x y-1)}
        \label{eq:qrt0H}
    \end{equation}
    which is an invariant for the map $\varphi_{1}$ \eqref{eq:qrt0}.
    This invariant was shown first in \cite{GJ_biquadratic}, where it
    was derived from the known invariant for the map \eqref{eq:rnahm112}
    using the birational conjugation with the map \eqref{eq:pj}.
    \label{rem:invariant}
\end{remark}

We have associated to $X_{1}$ the Picard group:
\begin{equation}
    \Pic (X_{1}) = \Span_{\Z}(H_{x}, H_{y}, E_{1},\dots,E_{8}).
    \label{eq:PicX1}
\end{equation}
The pullback of $\varphi_{1}$
act linearly on the Picard group of $X_{1}$ as follows:
\begin{equation}
    \varphi_{1}^{*}\colon
    \left(
        \begin{gathered}
            H_{x},H_{y},
            E_{1},E_{2},E_{3},
            \\
            E_{4},
            E_{5},E_{6},E_{7},E_{8} 
        \end{gathered}
    \right)\mapsto
    \left(
        \begin{gathered}
            H_{x}+H_{y}-E_{1}-E_{2},H_{x},
            E_{5},E_{6},H_{x}-E_{2},
            \\
            H_{x}-E_{1},
            E_{7},E_{8},E_{3},E_{4} 
        \end{gathered}
    \right).
    \label{eq:qrt0pb}
\end{equation}

Furthermore, the representatives on the Picard group of the irreducible
components of the singular fibres are given in the third column of
\Cref{tab:singoct}. Using the intersection form \eqref{eq:intform}, we can
find the associated surface type. This is given in the fourth and last
column of \Cref{tab:singoct}. This implies that the rational elliptic
fibration has a singular fibre configuration of type $(A_{3}\times A_{2}
\times 2A_{1})^{(1)}$, which is number 59 in Oguiso and Shioda's list
\cite[Table
8.2]{SchuttShioda2019mordell}.
This ends the proof of the first statement of \Cref{thm:mainA}.
The structure of the singular fibres is summarised in \Cref{fig:type59}.

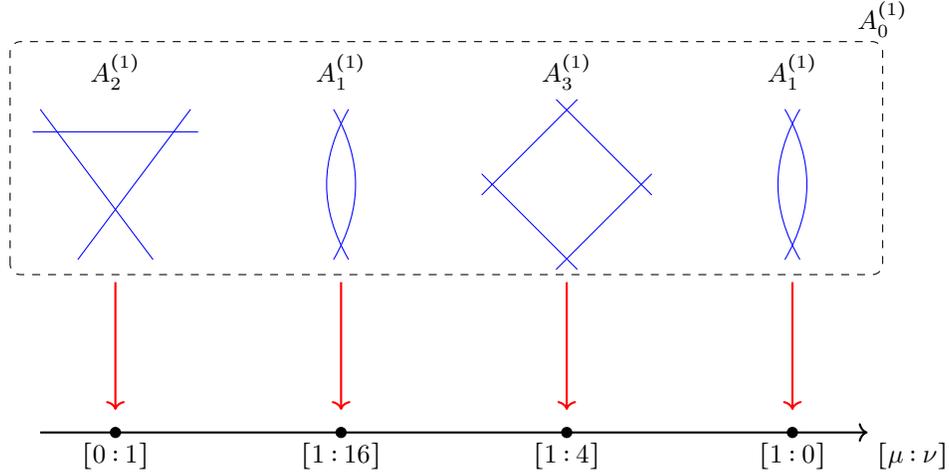
\begin{figure}[hbt]
    \centering
    \begin{tikzpicture}
        \draw[rounded corners,dashed] (-0.4, 2.1) rectangle (11.2, 5.2) {};
        \node at (11.2,5.5) {$A_{0}^{(1)}$};
        \draw[thick,->] (0,0) --(1,0) node[below] {$[0:1]$} 
            --(4,0) node[below] {$[1:16]$}
            --(7,0) node[below] {$[1:4]$}
            --(10,0) node[below] {$[1:0]$}
            -- (11,0) node[below right] {$[\mu:\nu]$};
        \node at (1,0) [circle,fill,inner sep=1.5pt]{};
        \node at (4,0) [circle,fill,inner sep=1.5pt]{};
        \node at (7,0) [circle,fill,inner sep=1.5pt]{};
        \node at (10,0) [circle,fill,inner sep=1.5pt]{};
        \draw[thick,->,red] (1,2) -- (1,0.3);
        \draw[thick,->,red] (4,2) -- (4,0.3);
        \draw[thick,->,red] (7,2) -- (7,0.3);
        \draw[thick,->,red] (10,2) -- (10,0.3);

        \draw[blue] (0.5,2.3) -- (2,4.3);
        \draw[blue] (1.5,2.3) -- (0,4.3);
        \draw[blue] (-0.1,4) -- (2.1,4);
        \node at (1,4.8) {$A_{2}^{(1)}$};

        \draw [blue] (3.9,2.3) to [bend right=30] (3.9,4.3);
        \draw [blue] (4.1,2.3) to [bend left=30] (4.1,4.3);
        \node at (4,4.8) {$A_{1}^{(1)}$};

        \begin{scope}[rotate around={45:(7,3.3)}]
            \draw[blue] (6.3,2.6)--(7.7,2.6)--(7.7,4)--(6.3,4)--cycle;
            \draw[blue] (6.3,2.4)--(6.3,2.6)--(6.1,2.6);
            \draw[blue] (7.7,2.4)--(7.7,2.6)--(7.9,2.6);
            \draw[blue] (7.9,4)--(7.7,4)--(7.7,4.2);
            \draw[blue] (6.3,4.2)--(6.3,4)--(6.1,4);
        \end{scope}
        \node at (7,4.8) {$A_{3}^{(1)}$};
        
        \draw [blue] (9.9,2.3) to [bend right=30] (9.9,4.3);
        \draw [blue] (10.1,2.3) to [bend left=30] (10.1,4.3);
        \node at (10,4.8) {$A_{1}^{(1)}$};
    \end{tikzpicture}
    \caption{A graphical representation of the singular fibres of
    the pencil $p_{1}$.}
    \label{fig:type59}
\end{figure}

The symmetry type $\Lambda^\perp$ of a map is the orthogonal complement
of the surface type in the Picard group. In particular, in the case
of $\Lambda=A_{2}^{(1)}$ we have $\Lambda^\perp=
E_{6}^{(1)}$. In \Cref{sec:octE6} we will show
that map $\varphi_{1}$ \eqref{eq:qrt0}, when considered on a $A_{2}^{(1)}$ type surface admits $A_{2}^{(1)}\subset E_6^{(1)}$ type
symmetry.

\subsection{The space of initial values of $\varphi_{2}$}

In this section, we prove the first statement of \Cref{thm:mainB} on the
existence of a covariant pencil with singular fibres of type $(3A_{3} \times A_{1})^{(1)}$ 
for the map $\varphi_{2}$ \eqref{eq:phi3}.
Again, we construct a space of initial values for such a map.
Following \Cref{sec:backgrownd} a space of initial values of the map
$\varphi_{2}$ \eqref{eq:phi3} is constructed by blowing up its singular
locus. By direct computations, we have:
\begin{equation}
    S_{2}:=\Sing \varphi_{2} = \left( c_{1},\dots, c_{7,8} \right)
    \label{eq:sing2}
\end{equation}
where:
\begin{equation}
    \begin{gathered}
    c_{1}\colon (x,y)=\left( -\frac{1+\imath\sqrt{3}}{2}, \frac{1-\imath\sqrt{3}}{2}\right),
    \quad
    c_{2} \colon (x,y)=\left(  \frac{1-\imath\sqrt{3}}{2}, -\frac{1+\imath\sqrt{3}}{2} \right),
    \\
    c_{3}\colon (x,y)=\left( -\frac{1-\imath\sqrt{3}}{2}, \frac{1+\imath\sqrt{3}}{2} \right),
    \quad
    c_{4}\colon (x,y)=\left( \frac{1+\imath\sqrt{3}}{2}, -\frac{1-\imath\sqrt{3}}{2}\right),
    \\
    c_{5}\colon (x,y)=\left(-\frac{1+\imath\sqrt{3}}{2}, -\frac{1-\imath\sqrt{3}}{2}  \right),
    \quad
    c_{6}\colon (x,y)=\left(-\frac{1-\imath\sqrt{3}}{2}, -\frac{1+\imath\sqrt{3}}{2}  \right),
    \\
    c_{7}\colon (x,y)=\left(1, 1 \right),
    \quad
    c_{8}\colon\left( x-1, \frac{y-1}{x-1} \right)=\left(0, -1 \right).   
    \end{gathered}
    \label{eq:bpphi3}
\end{equation}

We remark that not all the base points are immediately visible as
indeterminacy points of $\varphi_{2}$ \eqref{eq:phi3} and its inverse,
but some appear as a contraction of the exceptional lines associated
to a point after blowing up. 
We call $F_{i}$ the exceptional line associated
to the blowing up of the indeterminacy point $c_i$.
Then, $c_{i}$, $i=5,6$, appear as image of the exceptional lines $F_{1}$ and $F_{4}$
under $\varphi_{2}$ and $\varphi_{2}^{-1}$ respectively.  Finally, $c_{8}$
arises after blowing-up $c_{7}$ in the coordinate system specified in its
defining equation, i.e. $c_{7}$ and $c_{8}$ are infinitely near points.

Then, a space of initial values is $X_{2} = B_{\varphi_{2}}$,
a rational elliptic surface. A biquadratic pencil
$p_{2}$ is defined by the condition of passing
through the points $c_{i}$:
\begin{equation}
    \begin{aligned}
        p_{2}(x,y;\mu:\nu) &= 
        \mu(x^2 y^2 - 2 x^2 y -2 x y^2+4 x^2-2 x y+4 y^2+2 x-2 y+1)
        \\
        &+\nu(x y-1) (x y+1).
    \end{aligned}
    \label{eq:p2}
\end{equation}
By construction, this pencil is covariant with respect to the map
$\varphi_{2}$ \eqref{eq:phi3}.

Then, from \Cref{def:singular} we have that there are four singular fibres
corresponding to the four values of $\left[ \mu:\nu \right]$ shown in
the first column of \Cref{tab:singtetr}. The irreducible components of
the associated singular fibres are shown in the second column of the
same table.

\begin{table}[tp]
    \centering
    \rotatebox{90}{
    \begin{tabular}{cccc}
        \toprule
        Point & 
        Irred. comp. & 
        Repr. in $\Pic(X_{2})$ &
        Surf./symm. type
        \\
        $[\mu:\nu]$ & $d_{i} $ & $D_{i}$ & $\Lambda$/$\Lambda^\perp$
        \\
        \midrule
        \multirow{2}*{$\left[0:1\right]$}  &
        $\left\{xy+1 = 0\right\}$, & $H_x+H_y-F_1-F_2-F_3-F_4$
        & \multirow{2}*{$A_{1}^{(1)}$/$E_7^{(1)}$}
        \\
        & 
        $\left\{x y - 1=0   \right\}$, & $H_x+H_y-F_5-F_6-F_7-F_8$
        \\
        \midrule
        \multirow{3}*{$\left[1:0\right]$} & 
        $\left\{ xy-(1+\imath \sqrt{3})x-(1-\imath\sqrt{3})y+1=0 \right\} $ & $H_x+H_y-F_2-F_3-F_5-F_7$
        & \multirow{3}*{$A_{2}^{(1)}$/$E_6^{(1)}$}
        \\
        &
        $\left\{ xy-(1-\imath\sqrt{3})x-(1+\imath\sqrt{3})y+1=0 \right\} $ & $H_x+H_y-F_1-F_4-F_6-F_7$
        \\
        &
        $\left\{ \dfrac{y-1}{x-1}=0 \right\}\cup\left\{ \dfrac{x-1}{y-1}=0 \right\}$
        & $F_7-D_8$
        \\
        \midrule
        \multirow{3}*{$\left[-\imath \sqrt{3}, 9\right]$}  &
        $\left\{ y=-\frac{1+\imath \sqrt{3}}{2} \right\}$ & $H_y-F_2-F_6$, 
        & \multirow{3}*{$A_{2}^{(1)}$/$E_6^{(1)}$}
        \\
        & 
        $\left\{  x=-\frac{1+\imath \sqrt{3}}{2} \right\}$ & $H_x-F_1-E_5$
        \\
        & 
        $\left\{7xy-(2\imath\sqrt{3}+4)(x+y)+4\imath\sqrt{3}+1=0 \right\}$ & $H_x+H_y-F_3-F_4-F_7-F_8$
        \\
        \midrule
        \multirow{3}*{$\left[\imath \sqrt{3}, 9\right]$}  &
        $\left\{x=-\frac{1-\imath \sqrt{3}}{2}\right\}$ & $H_x-F_3-F_6$ & \multirow{3}*{$A_{2}^{(1)}$/$E_6^{(1)}$}
        \\
        &
        $\left\{ y=-\frac{1-\imath \sqrt{3}}{2} \right\}$ & $H_y-F_4-F_5$
        \\
        & 
        $\left\{ 7xy+(2\imath\sqrt{3}-4)(x+y)-4\imath\sqrt{3}+1=0 \right\}$ & $H_x+H_y-E_1-E_2-E_7-E_8$
        \\
        \bottomrule
    \end{tabular}
    }
    \caption{Irreducible components of the singular fibres in the pencil 
    associated to \eqref{eq:p2}, their representative in the Picard groups,
    and the associated surface and symmetry type.}
    \label{tab:singtetr}
\end{table}

\begin{remark}
    We remark that in the singular fibres presented in \Cref{tab:singtetr}
    the arrangement of the base points is different from the standard one
    presented e.g. in \cite{KajiwaraNoumiYamada2017R}. In particular,
    the presence of the infinitely near points $b_{7}$ and $b_{8}$
    implies that $X_{2}$ is not isomorphic to the elliptic fibration
    given in \cite{KajiwaraNoumiYamada2017R}.
    \label{rem:b78}
\end{remark}

\begin{remark}
    From the covariant pencil \eqref{eq:p2} we can build the rational 
    function:
    \begin{equation}
        \mathcal{H}_{2}
        =
        \frac{x^2 y^2 - 2 x^2 y -2 x y^2+4 x^2-2 x y+4 y^2+2 x-2 y+1}{(x y-1) (x y+1)},
        \label{eq:H2}
    \end{equation}
    which is an invariant for the map $\varphi_{2}$ \eqref{eq:phi3}.
    By direct computation it is possible to prove that this
    invariant arises from the one obtained from the construction in
    \cite{CelledoniMcLachlanMcLarenOwrenQuispel2017} for the map
    \eqref{eq:rnahm111} using the birational conjugation with the map
    \eqref{eq:pj3}.  
    \label{rem:invariant2}
\end{remark}

Associated to $X_{2}$ we have the Picard group:
\begin{equation}
    \Pic (X_{2}) = \Span_{\Z}( H_{x}, H_{y}, F_{1},\dots,F_{8}),
    \label{eq:PicX2}
\end{equation}
On the Picard group of $X_{2}$, the pullback of $\varphi_{2}$
act linearly as follows:
\begin{equation}
    \varphi_{2}^{*}\colon
    \left(
        \begin{gathered}
            H_{x},H_{y},
            \\
            F_{1},F_{2},F_{3},F_{4},
            \\
            F_{5},F_{6},F_{7},F_{8} 
        \end{gathered}
    \right)\mapsto
    \left(
        \begin{gathered}
            2 H_x+H_y-F_2-F_4-F_7-F_8, H_x,
            \\
            H_x-F_2, F_5, H_x-F_4, F_6,
            \\
            F_3, F_1, H_x-F_8, H_x-F_7
        \end{gathered}
    \right).
    \label{eq:phi3pb}
\end{equation}

Furthermore, the representatives on the Picard group of the irreducible
components of the singular fibres are shown in the third column of
\Cref{tab:singtetr}. Using the intersection form \eqref{eq:intform},
we can find the associated surface type. This is given in the fourth and
last column of \Cref{tab:singtetr}. This implies that the singular fibres
configuration is of type $(3A_{2}\times A_{1})^{(1)}$,
type 61 in Oguiso and Shioda's list
\cite[Table
8.2]{SchuttShioda2019mordell}. This ends of the proof of the first statement of
\Cref{thm:mainB}. The structure of the singular fibres is summarised
in \Cref{fig:type61}.

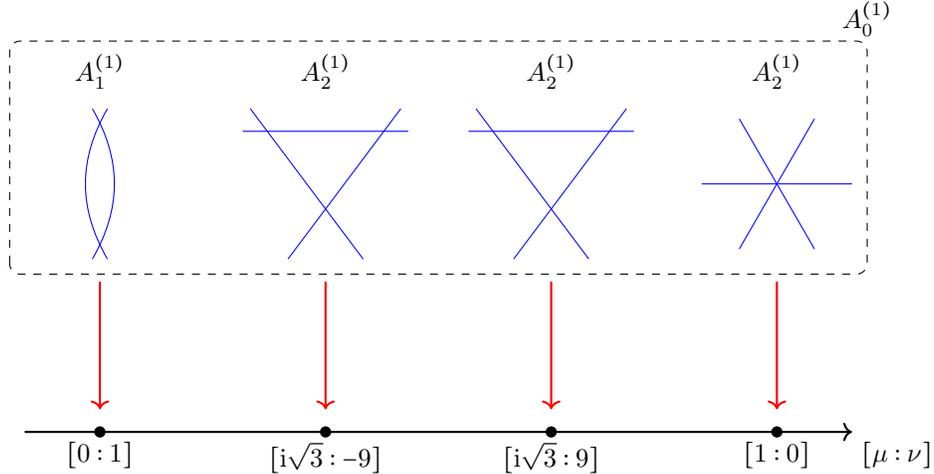
\begin{figure}[hbt]
    \centering
    \begin{tikzpicture}
        \draw[rounded corners,dashed] (-0.2, 2.1) rectangle (11.2, 5.2) {};
        \node at (11.2,5.5) {$A_{0}^{(1)}$};
        \draw[thick,->] (0,0) --(1,0) node[below] {$[0:1]$} 
            --(4,0) node[below] {$[\imath \sqrt{3}:-9]$}
            --(7,0) node[below] {$[\imath \sqrt{3}:9]$}
            --(10,0) node[below] {$[1:0]$}
            -- (11,0) node[below right] {$[\mu:\nu]$};
        \node at (1,0) [circle,fill,inner sep=1.5pt]{};
        \node at (4,0) [circle,fill,inner sep=1.5pt]{};
        \node at (7,0) [circle,fill,inner sep=1.5pt]{};
        \node at (10,0) [circle,fill,inner sep=1.5pt]{};
        \draw[thick,->,red] (1,2) -- (1,0.3);
        \draw[thick,->,red] (4,2) -- (4,0.3);
        \draw[thick,->,red] (7,2) -- (7,0.3);
        \draw[thick,->,red] (10,2) -- (10,0.3);

        \draw [blue] (0.9,2.3) to [bend right=30] (0.9,4.3);
        \draw [blue] (1.1,2.3) to [bend left=30] (1.1,4.3);
        \node at (1,4.8) {$A_{1}^{(1)}$};

        \draw[blue] (3.5,2.3) -- (5,4.3);
        \draw[blue] (4.5,2.3) -- (3,4.3);
        \draw[blue] (2.9,4) -- (5.1,4);
        \node at (4,4.8) {$A_{2}^{(1)}$};

        \draw[blue] (6.5,2.3) -- (8,4.3);
        \draw[blue] (7.5,2.3) -- (6,4.3);
        \draw[blue] (5.9,4) -- (8.1,4);
        \node at (7,4.8) {$A_{2}^{(1)}$};
        
        \draw[blue] (9,3.3) -- (11,3.3);
        \draw[blue] (10-0.5,3.3-0.866) -- (10+0.5,3.3+0.866);
        \draw[blue] (10-0.5,3.3+0.866) -- (10+0.5,3.3-0.866);
        \node at (10,4.8) {$A_{2}^{(1)}$};
    \end{tikzpicture}
    \caption{A graphical representation of the singular fibres of $p_{2}$.}
    \label{fig:type61}
\end{figure}

From \Cref{tab:singtetr} we see that 
the full symmetry of a system on a surface of type
$\Lambda=A_{1}^{(1)}$ is
$\Lambda^\perp=E_{7}^{(1)}$. Then, in \Cref{sec:B3} we show that map
$\varphi_{2}$ \eqref{eq:phi3}, when considered on a $A_{1}^{(1)}$ type
surface admits $B_{3}^{(1)}\subset E_7^{(1)}$ type symmetry.

\subsection{The map $\varphi_{s}$ and its initial value space}

In \Cref{rem:b78} we noticed that the point configuration of that
$A_1^{(1)}$ surface is different with respect to the standard one
considered for instance in~\cite{KajiwaraNoumiYamada2017R}. So, a natural
question to ask is if we can recover a more general pencil with a similar
structure, but no infinitely near points.

Consider a generic biquadratic polynomial and impose it to pass through 
the points $c_i$ with $i=1,\ldots,6$. By direct computation, is shown
that the obtained curve can be parametrised as:
\begin{equation}
\begin{aligned}
    p_s 
    &= \mu \left( x^{2}y^{2}-2x^{2}y-2 x y^{2}+4{x}^{2}-2 xy + 4 y^{2}-2x-2y+1 \right)
    \\
    &+\nu \left[ s \left( {x}^{2}+{y}^{2}+1\right) - \left( 1-s \right)  \left({x}^{2}{y}^{2}-1 \right) \right].
\end{aligned}
\label{eq:padd}
\end{equation}
In \cref{eq:padd} where $[\mu:\nu]\in\Pj^1$ are the parameters of the
pencil and $s$ is an additional parameter. For $s\neq0$ we have a pencil
with eight distinct base points.  For sake of simplicity we denote the
base points by $l_{i}$ and we list them as follows:
\begin{equation}
    \begin{gathered}
        l_{1}(s) \colon (x,y) =
        \left(\frac{(5 -3\sqrt{3} \imath) s^2 -32 s+32}{26 s^2-56 s+32},
        \frac{(5 s^2+ 3 \imath\sqrt{3}) s^2 -32 s+32}{26 s^2-56 s+32}\right), 
        \\
        l_{2}(s) \colon (x,y) =
        \left(\frac{(5 s^2+ 3 \imath\sqrt{3}) s^2 -32 s+32}{26 s^2-56 s+32},
        \frac{(5 -3\sqrt{3} \imath) s^2 -32 s+32}{26 s^2-56 s+32}\right)
        \\
        l_{3}=c_{5},\ l_{4}=c_{6},\ l_{5}=c_{2},\
        l_{6}=c_{4},\ l_{7}=c_{3},\ l_{8}=c_{1}.
    \end{gathered}
    \label{eq:bpl}
\end{equation}
When $s=0$ the two base points $l_{1}(s)$ and $l_{2}(s)$ coalesce and
we get the base points $c_{7,8}$. Note that, however the limit is 
singular.

We now define a map which is covariant on the pencil~\eqref{eq:padd}.
Since the pencil $p_{s}$~\eqref{eq:padd} is symmetric under the
involution $\iota\colon(x,y)\mapsto(y,x)$, we can adopt the
construction of the QRT root as explained 
in~\cite[\S 6.3]{HietarintaJoshiNijhoff2016}. Consider the 
rational function obtained by taking the coefficients of
the biquadratic polynomial $p_{s}$~\eqref{eq:padd} with
respect to $\mu$ and $\nu$:
\begin{equation}
    \mathcal{H}_{s} = 
    \frac{x^{2}y^{2}-2x^{2}y-2 x y^{2}+4{x}^{2}-2 xy + 4 y^{2}-2x-2y+1}{%
        s \left( {x}^{2}+{y}^{2}+1\right) - \left( 1-s \right)  \left({x}^{2}{y}^{2}-1 \right)}.
    \label{eq:Hs}
\end{equation}
Considering the map:
\begin{equation}
    \upsilon_{s}\colon (x,y)\mapsto (x,y'),
    \label{eq:ups}
\end{equation}
with $y'$ non-trivial solution of:
\begin{equation}
    \mathcal{H}_{s}(x,y') = \mathcal{H}_{s}(x,y).
    \label{eq:psprime}
\end{equation}
By direct computation from~\eqref{eq:psprime} we obtain a trivial 
solution $y'=y$, while the non-trivial one is:
\begin{equation}
    y' = \frac{3 s x^2 y+2 s x^2-3 s x y-4 x^2 y+s y+6 x y-4 y+2}{%
        2 s x^2 y-3 s x^2-2 x^2 y+3 s x+2 s y+4 x^2-s-6 x+4},
    \label{eq:yp}
\end{equation}
and the map $\upsilon_{s}$ is an involution. The composition 
$\iota\circ\upsilon_{s}$ is then explicitly given by the maps 
$\varphi_{s}$ of \Cref{eq:phis}. This is a map of infinite order, 
called a QRT root, is by construction covariant on 
$p_{s}$~\eqref{eq:padd} and invariant on 
$\mathcal{H}_{s}$~\eqref{eq:Hs}.

Now, by construction, the space of the initial conditions of the map
$\varphi_{s}$~\eqref{eq:phis}, namely $X_{s}=B_{\varphi_s}$, is 
constructed by blowing up $\PcrossP$ on the points $l_{i}$. Indeed, by
construction
\begin{equation}
    S_{s} = \Sing \varphi_{s} = \left( l_{1},\dots,l_{8} \right).
    \label{eq:Sings}
\end{equation}
Again, by construction $X_{s}$ is a rational elliptic surface. We 
denote the exceptional lines corresponding to the blow up of the 
points $l_{i}$ by $L_{i}$.

Associated to $X_{s}$ we have the Picard group: 
\begin{equation}
    \Pic (B_{\varphi_s}) = \Span_{\Z} (H_x,H_y,L_1,\ldots,L_8).
\end{equation}
Then the action on $\Pic (B_{\varphi_s})$ is given by:
\begin{equation}
\varphi^{*}_s \colon
\begin{pmatrix}
H_x,H_y,
\\
L_1,L_2,L_3,L_4,
\\
L_5,L_6,L_7,L_8
\end{pmatrix}
\mapsto
\begin{pmatrix}
2 H_x+ H_y-L_1-L_2- L_5- L_6, H_x,
\\ 
H_x - L_2, H_x - L_1, L_7, L_8,
\\ 
L_3, L_4, H_x- L_6, H_x - L_5
\end{pmatrix}.
\end{equation}

Then, from \Cref{def:singular} we find that the pencil 
$p_s$~\eqref{eq:padd} has five singular fibres corresponding to the
following values of $[\mu:\nu]\in\Pj^{1}$:
\begin{subequations}
    \begin{align}
        \zeta_{1} &= [1:0],
        \\
        \zeta_{2}(s) &= [-s/4, 1],
        \\
        \zeta_{3}(s) &= [\imath( s-2) \sqrt{3}-3 s, 18],
        \\
        \zeta_{4}(s) &= [-\imath(s-2)\sqrt{3}-3 s, 18],
        \\
        \zeta_{5}(s) &= [(1-s) (s^2-s+1), s^3].
    \end{align}
    \label{eq:singfphis}%
\end{subequations}
This information, along with the irreducible components of the 
associated singular fibres are shown in the first and second column 
of \Cref{tab:sings}.

\begin{remark}
    In the form of the singular values of the pencil 
    $p_{s}$~\eqref{eq:padd} we choose to underline that the
    point $\zeta_{1}$ is independent from $s$, while $\zeta_{i}(s)$, 
    for $i=2,3,4,5$ the dependence is explicit. The reason will more
    evident from the discussion we will make in \Cref{rem:degfibr}.
    \label{rem:nos}
\end{remark}

\begin{table}[hb]
    \centering
    \rotatebox{90}{
    \begin{tabular}{cccc}
        \toprule
        Point & 
        Irred.\ comp. & 
        Repr.\ in $\Pic(X_{2})$ &
        Surf./symm.\ type
        \\
        $[\mu:\nu]$ & $d_{i} $ & $D_{i}$ & $\Lambda$/$\Lambda^\perp$
        \\
        \midrule
        \multirow{2}*{$\zeta_{1}$}  &
        $\left\{xy-(1+\imath \sqrt{3})x-(1-\imath\sqrt{3})y+1=0\right\}$, & 
        $H_{x}+H_{y} - L_{1} -L_{3} -L_{5} -L_{7}$
        & \multirow{2}*{$A_{1}^{(1)}$/$E_7^{(1)}$}
        \\
        & 
        $\left\{xy-(1-\imath\sqrt{3})x-(1+\imath\sqrt{3})y+1=0\right\}$, & 
        $H_{x}+H_{y} - L_{2} -L_{4} -L_{6} -L_{8}$
        \\
        \midrule
        \multirow{2}*{$\zeta_{2}(s)$} & 
        $\left\{\displaystyle xy+2 \frac {sx}{3 s-4}+2 \frac {sy}{3 s-4} =\frac {s-4}{3 s-4} \right\} $ & 
        $H_{x}+H_{y} - L_{1} -L_{2} -L_{3} -L_{4}$
        & \multirow{2}*{$A_{2}^{(1)}$/$E_6^{(1)}$}
        \\
        &
        $\left\{xy+1=0 \right\} $ & $H_{x}+H_{y} - L_{5} -L_{6} -L_{7} -L_{8}$
        \\
        \midrule
        \multirow{3}*{$\zeta_{3}(s)$}  &
        $\left\{\displaystyle x=-\frac{1-\imath\sqrt {3}}{2} \right\}$ & 
        $H_{x} - L_{4} -L_{7}$, 
        & \multirow{3}*{$A_{1}^{(1)}$/$E_7^{(1)}$}
        \\
        & 
        $\left\{ \displaystyle y=-\frac{1-\imath\sqrt {3}}{2} \right\}$ & 
        $H_{y} - L_{3} -L_{6}$
        \\
        & 
        $\left\{ \displaystyle
            \begin{gathered}
                xy
                +\frac{1}{14} 
                \frac{\left(13\imath\sqrt {3}-5 \right)\left( 4 \imath\sqrt{3}-7 s+8 \right) x}{\imath\sqrt {3}-19 s+ 23}
                \\
                +\frac{1}{14} {\frac { \left( 13 \imath\sqrt {3}-5 \right)  \left( 4 \imath\sqrt {3}-7 s+8 \right) y}{\imath\sqrt {3}-19 s+23}}
                -\frac{1}{2} 
                \frac{ \left( 5 \imath\sqrt {3}+1 \right)  \left( \imath\sqrt {3}-s+5 \right) }{\imath\sqrt {3}-19 s+23}
            \end{gathered}
        \right\}$ & 
        $H_{x}+H_{y}-L_{1}-L_{2}-L_{5}-L_{8}$
        \\
        \midrule
        \multirow{3}*{$\zeta_{4}(s)$}  &
        $\left\{ \displaystyle x=-\frac{1+\imath\sqrt {3}}{2} \right\}$ & 
        $H_{x} - L_{3} -L_{8}$ & \multirow{3}*{$A_{2}^{(1)}$/$E_6^{(1)}$}
        \\
        &
        $\left\{ \displaystyle y=-\frac{1+\imath\sqrt {3}}{2} \right\}$ & 
        $H_{y} - L_{4} -L_{5}$
        \\
        & 
        $\left\{\displaystyle
            \begin{gathered}
                xy 
                -\frac{1}{14} {\frac { \left( 13 \imath\sqrt {3}+5 \right)  \left( 4 \imath\sqrt {3}+7 s-8 \right) x}{\imath\sqrt {3}+19 s- 23}}
                \\
                -\frac{1}{14} {\frac { \left( 13 \imath\sqrt {3}+5 \right)  \left( 4 \imath\sqrt {3}+7 s-8 \right) y}{\imath\sqrt {3}+19 s-23}} 
                + \frac{1}{2} {\frac { \left( 5 \imath\sqrt {3}-1 \right)  \left( s+\imath\sqrt {3}-5 \right) }{\imath\sqrt {3}+19 s-23}}
= 0
            \end{gathered}
        \right\}$ & 
        $H_{x}+H_{y}-L_{1}-L_{2}-L_{6}-L_{7}$
        \\
        \midrule
        $\zeta_{5}(s)$ & $\left\{ p_{5}^{*}=0 \right\}$ & $-K_{X_{s}}$ & $A_{0}^{(1)*}$/$E_{8}^{(1)}$
        \\
        \bottomrule
\end{tabular}}
    \caption{Irreducible components of the singular fibres in the 
        pencil associated to~\eqref{eq:padd}, their representative in 
        the Picard groups, and the associated surface and symmetry 
        type. The biquadratic polynomial $p_{5}^{*}$ is given in
    \Cref{eq:p5star}.}
    \label{tab:sings}
\end{table}

Furthermore, the representatives on the Picard group of the 
irreducible components of the singular fibres are shown in the third 
column of \Cref{tab:sings}. Using the intersection 
form~\eqref{eq:intform} we can find the associated surface type. This 
is given in the fourth and last column of \Cref{tab:sings}. This 
implies that the singular fibre structure is 
${(2A_2\times 2 A_1 \times A_0^{*})}^{(1)}$, that is case 40 
from~\cite[Table 8.2]{SchuttShioda2019mordell}\footnote{Note that in 
    that classification the contribution from $A_0^{(1)*}$ fibre is not present since 
it does not contribute to the Mordell--Weil lattice.}. 
The fibre of type $A_{0}^{(1)*}$ has as representative the following
irreducible curve in $\PcrossP$ the nodal biquadric:
\begin{equation}
    \begin{aligned}
        p_{5}^{*} &=
        \left( {s}^{2}+1 \right)  \left( s-1 \right)^{2}{x}^{2}{y}^{2}
        +2 \left( s-1 \right)  \left( {s}^{2}-s+1 \right) {x}^{2}y
        \\
        &+ \left( {s}^{2}-2 s+2 \right)^{2}{x}^{2}
        +2  \left( s-1 \right)\left( {s}^{2}-s +1 \right) x{y}^{2}
        \\
        &+2  \left( s-1 \right)  \left( {s}^{2}-s+1 \right) xy
        +2  \left( s-1 \right)  \left( {s}^{2}-s+1 \right) x
        \\
        &+\left( {s}^{2}-2 s+2 \right)^{2}{y}^{2}
        +2  \left( s-1 \right)  \left( {s}^{2}-s+1 \right) y+2 {s}^{2}-2 s+1,
    \end{aligned}
    \label{eq:p5star}
\end{equation}
having a node in the point:
\begin{equation}
    n_{s}\colon (x,y) = \left(\frac{1}{1-s},\frac{1}{1-s}  \right).
    \label{eq:p5sing}
\end{equation}
The fibre structure is graphically represented in \Cref{fig:type40}. 
The ends the proof of the first statement of \Cref{thm:mainC}.

\begin{remark}
    We remark that as $s\to0$ we have that the singular fibre 
    $\zeta_{5}(s)$ collide with the fibre $\zeta_{1}=[0:1]$. That 
    is, they originate the second fibre of the map $\varphi_{2}$, 
    listed in \Cref{tab:singtetr}: the collision of a fibre of 
    surface type $A_{0}^{(1)*}$ and $A_{1}^{(1)}$ generates a fibre 
    of surface type $A_{3}^{(1)}$. This is because both the 
    irreducible components of the fibre in $\zeta_{1}$, see 
    \Cref{tab:sings} pass through the point $(x,y)=(1,1)=b_{7}$ which 
    becomes a base point for the pencil as $s\to0$. Note also that
    the nodal point $n_{s}$~\eqref{eq:p5sing} of the curve 
    $p_{5}^{*}$~\eqref{eq:p5star} tend to new base point $c_{7}$. The 
    fibre structure is then changed to $A_{3}^{(1)}$ because the 
    proper transform of the irreducible components of the fibre in
    $\zeta_{1}$ and the exceptional line on which $c_{7,8}$ lie all
    intersect on a single point, see \Cref{fig:type61}.
    \label{rem:degfibr}
\end{remark}

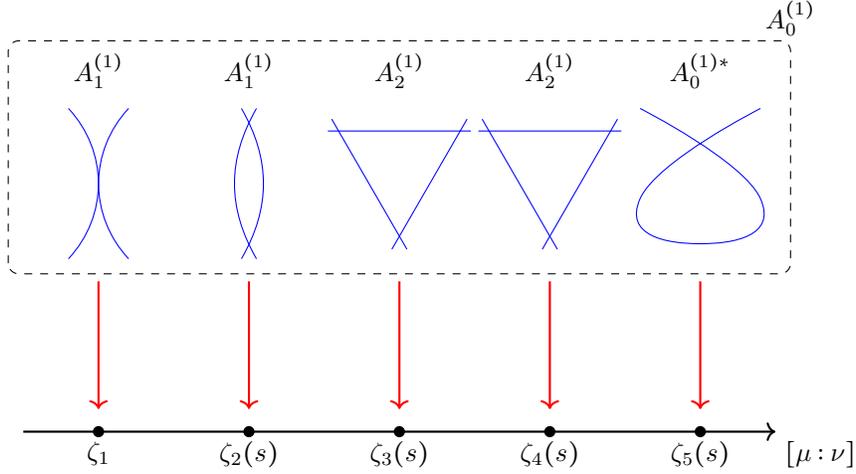
\begin{figure}[htb]
    \centering
    \begin{tikzpicture}
        \draw[rounded corners,dashed] (-0.2, 2.1) rectangle (10.2, 5.2) {};
        \node at (10.2,5.5) {$A_{0}^{(1)}$};
        \draw[thick,->] (0,0) --(1,0) node[below] {$\zeta_{1}$} 
            --(3,0) node[below] {$\zeta_{2}(s)$}
            --(5,0) node[below] {$\zeta_{3}(s)$}
            --(7,0) node[below] {$\zeta_{4}(s)$}
            -- (9,0) node[below] {$\zeta_{5}(s)$}
            -- (10,0) node[below right] {$[\mu:\nu]$};
        \node at (1,0) [circle,fill,inner sep=1.5pt]{};
        \node at (3,0) [circle,fill,inner sep=1.5pt]{};
        \node at (5,0) [circle,fill,inner sep=1.5pt]{};
        \node at (7,0) [circle,fill,inner sep=1.5pt]{};
        \node at (9,0) [circle,fill,inner sep=1.5pt]{};
        \draw[thick,->,red] (1,2) -- (1,0.3);
        \draw[thick,->,red] (3,2) -- (3,0.3);
        \draw[thick,->,red] (5,2) -- (5,0.3);
        \draw[thick,->,red] (7,2) -- (7,0.3);
        \draw[thick,->,red] (9,2) -- (9,0.3);

        \draw [blue] (0.6,2.3) to [bend right=43] (0.6,4.3);
        \draw [blue] (1.4,2.3) to [bend left=43] (1.4,4.3);
        \node at (1,4.8) {$A_{1}^{(1)}$};

        \draw [blue] (2.9,2.3) to [bend right=30] (2.9,4.3);
        \draw [blue] (3.1,2.3) to [bend left=30] (3.1,4.3);
        \node at (3,4.8) {$A_{1}^{(1)}$};

        \draw[blue] (5-0.2*0.5,2.6-0.2*0.866) -- (5+1.8*0.5,2.6+1.8*0.866);
        \draw[blue] (5-1.8*0.5,2.6+1.8*0.866) -- (5+0.2*0.5,2.6-0.2*0.866);
        \draw[blue] (5-1.9*0.5,4) -- (5+1.9*0.5,4);
        \node at (5,4.8) {$A_{2}^{(1)}$};
        
        \draw[blue] (7-0.2*0.5,2.6-0.2*0.866) -- (7+1.8*0.5,2.6+1.8*0.866);
        \draw[blue] (7-1.8*0.5,2.6+1.8*0.866) -- (7+0.2*0.5,2.6-0.2*0.866);
        \draw[blue] (7-1.9*0.5,4) -- (7+1.9*0.5,4);
        \node at (7,4.8) {$A_{2}^{(1)}$};

        \draw [blue] plot [smooth, tension=1] coordinates {(8.2,4.3) (9.8,3.1) (9,2.5) (8.2,3.1) (9.8,4.3)};

        \node at (9,4.8) {$A_{0}^{(1)*}$};
    \end{tikzpicture}
    \caption{A graphical representation of the singular fibres of $p_{s}$.}
    \label{fig:type40}
\end{figure}

\section{Singular fibre of type $A_{2}^{(1)}$ associated with the map $\varphi_{1}$}
\label{sec:octE6}

In this section we compute the action of the map
$\varphi_{1}$~\eqref{eq:qrt0} on the orthogonal complement of the singular
fibre $[0:1]$ of type $A_{2}^{(1)}$ of pencil $p_{1}$~\eqref{eq:p1}. Then,
we compute the normalizer of the Weyl group $W(4A_{1})$ in the extended
Weyl group $\widetilde{W}( E_{6}^{(1)})$. This will yield that the map
$\varphi_{1}$~\eqref{eq:qrt0} on such a singular fibre has symmetry
type $A_{2}^{(1)}\subset E_{6}^{(1)}$, thus ending the proof of
\Cref{thm:mainA}.

\subsection{Action on the simple roots of $E_{6}^{(1)}$ and decomposition in $\widetilde{W}(E_6^{(1)})$}

Following Table \ref{tab:singoct}, we have that the singular fiber of the 
pencil \eqref{eq:p1} at the point $\left[ \mu:\nu \right]=\left[ 0:1\right]$ 
is of type $A_2^{(1)}$ and is generated by the following divisors in $\Pic(X_1)$:
\begin{equation}
    D_{0} = H_{y}-E_{5}-E_{6},
    \quad
    D_{1} = H_{x}-E_{7}-E_{8},
    \quad
    D_{2} = H_{x}+H_{y}-E_{1}-E_{2}-E_{3}-E_{4}.
    \label{eq:qrt0Di}
\end{equation}
The three divisors are permuted by the pullback map $\varphi^{*}_{1}$, that is
$\varphi_{1}^{*}\left( D_{i} \right)=D_{i+1}$, $i\in \Z/3\Z$. The rational
elliptic surface $X_1$ with this choice of the singular fibre is shown
in \Cref{fig:pic}.

The orthogonal complement of $A_{2}^{(1)}$ type singular fiber in 
$\Pic\left( X_1 \right)$ is generated by:
\begin{equation}
    \begin{gathered}
    \alpha_{0} = E_5-E_6,
    \quad
    \alpha_{1} = E_8-E_7,
    \quad
    \alpha_{2} = H_y-E_1-E_8,
    \quad
    \alpha_{3} = E_1-E_2, 
    \\
    \alpha_{4} = E_2-E_3,
    \quad
    \alpha_{5} = E_3-E_4,
    \quad
    \alpha_{6} =H_x-E_1-E_5.
    \end{gathered}
    \label{eq:E61}
\end{equation}

The set $\De^{(1)}=\De\left(E_6^{(1)}\right)=\{\al_j\mid 0\leq j\leq 6\}$ is a simple system of type
$E_6^{(1)}$, its Dynkin diagram is given in \Cref{fig:e6}. The anti-canonical divisor corresponds to the null root
of the $E_{6}^{(1)}$ system:
\begin{equation}
    -K_{X_{1}}=\delta = 
    \alpha_{015224466333}.
    \label{eq:E61nullroot}
\end{equation}
Let $V^{(1)}$ be a vector space
with basis 
$\De^{(1)}$,
equipped with a symmetric bilinear form given by Equations \eqref{alaij0}
and \eqref{CarE6a}. The corresponding extended affine Weyl group is, 
\begin{equation}
\widetilde{W}(E_{6}^{(1)})=\langle s_i\;|\;0\leq i \leq
6\rangle\rtimes\langle \sigma\rangle
=W(E_{6})\ltimes P=W(E_{6})\ltimes\lan u_j\, \mid 1\leq j\leq 6\ran,
\end{equation}
where $u_j$ is the translation by the fundamental weight $h_j$ $(1\leq j\leq 6)$
of the $E_6^{(1)}$ system.
The action of 
$s_j\in W(E_6^{(1)})$ on $\De^{(1)}$ is given by Equations
\eqref{sij0} and \eqref{CarE6a}.
See \Cref{app:E6} for
basic data and properties of $\widetilde{W}(E_{6}^{(1)})$.

\begin{figure}[ht!]
    \centering
    \begin{tikzpicture}[scale=2]
        \def\pt{2};
        \node[green,above,left] at (0,2) {$H_{x}-E_{7}-E_{8}$};
        \draw[thick,->,green] (0,-2)--(0,2);
        \node[green,below,right] at (2,0) {$H_{y}-E_{5}-E_{6}$};
        \draw[thick,->,green,dashed] (-2,0)--(2,0);
        \draw[domain=1/2:2, smooth, variable=\x, green,thick,dash dot] plot (\x, 1/\x);
        \draw[domain=-2:-1/2, smooth, variable=\x, green,thick,dash dot] plot (\x, 1/\x);
        \node[green,above,right] at (1/2,2) {$H_{x}+H_{y}-E_{1}-E_{2}-E_{3}-E_{4}$};
        \node[red,right] at ($({sqrt(\pt)+1/6},{1/sqrt(\pt)+1/3})$) (E1) {$E_{1}$};
        \draw[red] ($({sqrt(\pt)-1/6},{1/sqrt(\pt)-1/3})$)--($({sqrt(\pt)+1/6},{1/sqrt(\pt)+1/3})$) ;
        \node[red,left] at ($({-sqrt(\pt)-1/6},{-1/sqrt(\pt)-1/3})$) (E2) {$E_{2}$};
        \draw[red] ($({-sqrt(\pt)-1/6},{-1/sqrt(\pt)-1/3})$) --($(-{sqrt(\pt)+1/6},{-1/sqrt(\pt)+1/3})$) ;
        \node[red,below] at ($({-1/sqrt(\pt)-1/3},{-sqrt(\pt)-1/6})$) (E3) {$E_{3}$};
        \draw[red] ($({-1/sqrt(\pt)-1/3},{-sqrt(\pt)-1/6})$) --($({-1/sqrt(\pt)+1/3},{-sqrt(\pt)+1/6})$) ;
        \node[red,below] at ($({1/sqrt(\pt)+1/3},{sqrt(\pt)+1/6})$) (E4) {$E_{4}$};
        \draw[red] ($({1/sqrt(\pt)-1/3},{sqrt(\pt)-1/6})$) --($({1/sqrt(\pt)+1/3},{sqrt(\pt)+1/6})$) ;
        \node[red,right] at ($({1/sqrt(\pt)},{1/3})$) (E5) {$E_{5}$};
        \draw[red] ($({1/sqrt(\pt)},{-1/3})$) --($({1/sqrt(\pt)},{1/3})$) ;
        \node[red,left] at ($({-1/sqrt(\pt)},{-1/3})$) (E6) {$E_{6}$};
        \draw[red] ($({-1/sqrt(\pt)},{-1/3})$) --($({-1/sqrt(\pt)},{1/3})$) ;
        \node[red,below] at ($({1/3},{-1/sqrt(\pt)})$) (E7) {$E_{7}$};
        \draw[red] ($({-1/3},{-1/sqrt(\pt)})$) --($({1/3},{-1/sqrt(\pt)})$) ;
        \node[red,above] at ($({-1/3},{1/sqrt(\pt)})$) (E8) {$E_{8}$};
        \draw[red] ($({-1/3},{1/sqrt(\pt)})$) --($({1/3},{1/sqrt(\pt)})$) ;
    \end{tikzpicture}
    \caption{The space of initial values of \eqref{eq:qrt0} interpreted
        as elliptic surface of type $A_{2}^{(1)}$.}
    \label{fig:pic}
\end{figure}
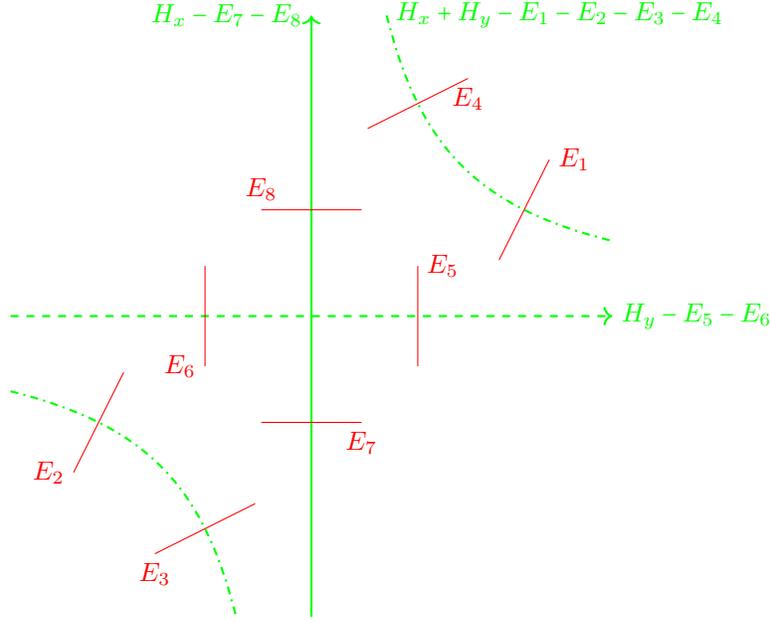

\begin{figure}[ht!]
    \centering
    \begin{tikzpicture}[thick]
        \node (a1) at (0,0){$\alpha_{1}$};
        \node (a2) at (1,0) {$\alpha_{2}$};
        \node (a3) at (2,0) {$\alpha_{3}$};

        \node (a4) at (3,0) {$\alpha_{4}$};
        \node (a5) at (4,0) {$\alpha_{5}$};
        \node (a6) at (2,1) {$\alpha_{6}$};
        \node (a0) at (2,2) {$\alpha_{0}$};
        \draw[-] (a1) -- (a2);
        \draw[-] (a2) -- (a3);
        \draw[-] (a3) -- (a4);
        \draw[-] (a4) -- (a5);
        \draw[-] (a0) -- (a6);
        \draw[-] (a3) -- (a6);
        \centerarc[dashed,<-](2,0)(-5:-175:1);
        \centerarc[dashed,<-](2,0)(-5:-175:2);
        \centerarc[dashed,->](2,0)(10:75:1);
        \centerarc[dashed,->](2,0)(5:80:2);
        \centerarc[dashed,->](2,0)(105:165:1);
        \centerarc[dashed,->](2,0)(100:170:2);
        \node (pi2) at ({2+2.3*cos(50)},{2.3*sin(50)}) {$\sigma$};
    \end{tikzpicture}
    \caption{$\Ga(E_{6}^{(1)})$ with diagram 
    automorphism $\sigma=(150)(246)$.}
    \label{fig:e6}
\end{figure}
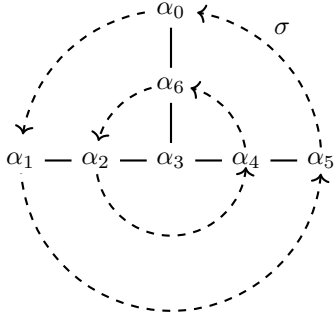

Using Equations \eqref{eq:qrt0pb} and \eqref{eq:E61}
we have $\varphi^{*}_{1}$ acts on $\De^{(1)}$ by:
\begin{equation}
    \varphi^{*}_{1}
    \colon
    \left( 
        \alpha_{0},\alpha_{1},\alpha_{2},\alpha_{3},
        \alpha_{4},\alpha_{5},\alpha_{6}
    \right)
    \mapsto
    \left( 
        -\alpha_1, -\alpha_5, \alpha_{3456}, \alpha_0,
        -\alpha_{036}, \alpha_3, \alpha_{1236}
    \right).
    \label{eq:qrt0E61}
\end{equation}
Computing four iterations of the map $\varphi_{1}^{*}$ we obtain:
\begin{equation}
   \left(\varphi_{1}^{*}\right)^{4}
    \colon
    \left( 
        \alpha_{0},\alpha_{1},\alpha_{2},\alpha_{3},
        \alpha_{4},\alpha_{5},\alpha_{6}
    \right)
    \mapsto
    \left( 
        \alpha_{0},\alpha_{1},\alpha_{2},\alpha_{3},
        \alpha_{4}-\de,\alpha_{5},\alpha_{6}+\de
    \right).
    \label{eq:qrt0E614}
\end{equation}
By Equation \eqref{uja} we recognise that
$\left(\varphi_{1}^{*}\right)^{4}\in P$ is
the translation
by $h_4-h_6$, to which
$\varphi_{1}^{*}$ is a quasi-translation 
of order four in $\widetilde{W}\left(E_6^{(1)}\right)$.
It can be checked that $|h_4-h_6|^2=\frac{4}{3}$ using \Cref{hijE6}.
Finally, $\varphi_{1}^{*}$ has the following decomposition
in terms of the generators of $\widetilde{W}( E_{6}^{(1)} )$ using Equations \eqref{lfun} and \eqref{uaw}:
\begin{equation}
    \varphi_{1}^{*} = \sigma s_{453401}.
    \label{eq:qrt0dec}
\end{equation}

\subsection{The normalizer of ${W}(4A_1)$ in $\widetilde{W}( E_{6}^{(1)} )$}\label{Ne6}

The initial observation on $\varphi^{*}_{1}$ comes from its actions on $\De^{(1)}$
given in Equation \eqref{eq:qrt0E61}, that is $\varphi^{*}_{1}$ acts on the subset  
$\{\al_0, \al_1, \al_3, \al_5\}\subset \De^{(1)}$ as a permutation of order four modulo the minus signs. This prompts us to take
\begin{equation}\label{Je6}
J=\{\al_0, \al_1, \al_3, \al_5\}\cong 4A_1, 
\end{equation}
and consider the normalizer of 
$W_J=\lan s_0, s_1, s_3, s_5\ran$ in $\widetilde{W}(E_6^{(1)})$.
\begin{proposition}\label{NJWJe6}
  The normalizer of ${W}_J$ with $J=\{\al_0, \al_1, \al_3, \al_5\}$ in $\widetilde{W}( E_{6}^{(1)} )$ is given by
\begin{equation}\label{NWJe6}
N(W_J)=N_J\ltimes W_J=\left(\lan b_1, b_2, b_0\ran\rtimes\lan \sigma\ran\right)\ltimes
\lan s_0, s_1, s_3, s_5\ran
\cong \widetilde{W}\left(A_2^{(1)}\right)\ltimes {W}(4A_1),
\end{equation}  
with
\begin{equation}\label{A2sr}
    b_1 = s_{4534},\quad
    b_2 = s_{6036}, \quad
    b_0 = s_{2312}.
\end{equation}
\end{proposition}
\begin{proof}
 Recall that $N(W_J)=N_J\ltimes W_J$, 
where $N_J$ generated by the R- and M-elements,
is the group of all elements of $\widetilde{W}(E_{6}^{(1)})$
that act permutatively on the set $J$. Here we have 
$J=\{\al_0, \al_1, \al_3, \al_5\}\cong 4A_1$, and
$W_J=\lan s_0, s_1, s_3, s_5\ran\cong {W}(4A_1)$. 
Following the procedure of computing normalizers given in
\cite{H, BH}, we find that the
R-elements (that is, quasi-reflections) are given by,
\begin{equation}
    \label{A2srp}
    b_1 = s_{4534},\quad
    b_2 = s_{6036}, \quad
    b_0 = s_{2312},
\end{equation}
where $b_1,b_2,b_0$ act on the subset $J$ as permutations of order two,
written
as a permutation on the index set of $J$ we have,
\begin{equation}
    b_1=(35),
    \quad
    b_2=(03),
    \quad
    b_0=(13).
\end{equation}
The M-element is given by $\sigma$ 
(action on $\De^{(1)}$ in Equation \eqref{4ddae6}), 
that is its action on $J$ is,
\begin{equation}\label{4ddae6J}
    \sigma=(150).
\end{equation}
It can be verified using the defining relations of $\widetilde{W}(E_6^{(1)})$
that $b_i$'s are involutions that satisfy the defining relations
\eqref{funWA2} for an affine Weyl group of type $A_2^{(1)}$, while
$\sigma$ acts on the $b_i$'s by Equation \eqref{funWaA2}.
\begin{subequations}
    \begin{align}
        \label{funWA2}
        &b_j^2=1,\;\;(b_jb_{j+1})^3=1,\;\;(j\in \mb Z/3\mb Z),
        \\
        \label{funWaA2}
        &\sigma^3=1, \;\;\sigma b_j=b_{j+1}\sigma.
    \end{align}
\end{subequations}
That is, together the R- and M-elements generate $N_J$,
an extended affine Weyl group of type $A_2^{(1)}$,
\begin{equation}\label{nae6}
N_J=\lan b_1, b_2, b_0\ran\rtimes\lan \sigma\ran\cong \widetilde{W}(A_2^{(1)}).
\end{equation}   
\end{proof}
Consider the expression for $\varphi^{*}_{1}$ in Equation \eqref{eq:qrt0dec}, we see that $\varphi_{1}^{*} = 
\sigma s_{4534}s_{0} s_{1}=\sigma b_1s_{0} s_{1}$
is certainly an element of the subgroup
$N(W_J)$ of  $\widetilde{W}\left(E_6^{(1)}\right)$
given by \Cref{NWJe6}.

In what follows we further clarify the nature of $\varphi^{*}_{1}$, 
showing that there is a conjugation that takes
$\varphi^{*}_{1}$ to an element of
quasi-translation in $N(W_J)$.

\subsubsection{Sub-root system and translations in
$N_J\cong\widetilde{W}( A_{2}^{(1)} )$}
\label{Te6}

First we find a subspace $V_J^\perp\subset V^{(1)}$ on which $b_i\in N_J$
$(0\leq i \leq 2)$ can be realised as simple reflections of $A_2^{(1)}$
type. We then construct an element of translation in $\widetilde{W}( A_{2}^{(1)}
)$ and find its relation to
the $\varphi^{*}_{1}$ given in \Cref{eq:qrt0dec}.

\begin{proposition}
    The root system for $N_J$ given in Equation \eqref{nae6} is generated by 
    $\be=\{\be_i\mid 0\leq i \leq 2\}$ with
    \begin{align}\label{bre6}
        \be_0&=\al_2+\frac{\al_{13}}{2},\quad
        \be_1=\al_4+\frac{\al_{35}}{2},\quad
        \be_2=\al_6+\frac{\al_{03}}{2},\quad
    \end{align} 
    is of $A_2^{(1)}$ type, and  $V_J^\perp=\Span\left(\be\right)$. 
    The group $N_J=\lan b_1, b_2, b_0\ran\rtimes\lan\sigma\ran$ can be realised as 
    $\widetilde{W}(A_2^{(1)})$ on $V_J^\perp$ with $b_i$ acting as the 
    reflection along the root $\be_i$, and $\sigma$ is the diagram 
    automorphism of  $\Ga(A_2^{(1)})$ for the $\be$-system, see Figure \ref{bA2e6}.
\end{proposition}

\begin{proof}
Let $\be=\{\be_i\mid 0\leq i \leq 2\}$ with the $\be_i$'s
be as given in Equation \eqref{bre6}.
The bilinear form on $V^{(1)}$ in the $\De^{(1)}$ basis given by Equations \eqref{alaij0} and \eqref{CarE6a} is used to check that the $\be$-system
is or dimension three and orthogonal to $V_J$. That is, $V_J^\perp=\Span\left(\be\right)$.
Moreover, we found that $|\be_i|^2=2$,
so $\oc\be_j=\be_j$ for $(0\leq i \leq 2)$.

Computing the values of $(\be_i\cdot\oc\be_j)$ for the $\be$-system \eqref{bre6} again using the bilinear form on $V^{(1)}$ 
we have,
\begin{equation}\label{CarA2a}
 C(A_2^{(1)})=(\be_i\cdot\oc\be_j)_{1\leq i,j\leq 2, 0}=(\be_i\cdot\be_j)_{1\leq i,j\leq 2, 0}=\left(
\begin{array}{ccc}
 2 & -1 & -1 \\
 -1 & 2 & -1 \\
 -1 & -1 & 2 \\
\end{array}
\right).
\end{equation}
The matrix in Equation \eqref{CarA2a}
is the Cartan matrix of $A_2^{(1)}$ type.
That is, $\be$-system is of $A_2^{(1)}$ type, its
Dynkin diagram given in Figure \ref{bA2e6}.
The null root of this $A_2^{(1)}$ system is given by
\begin{equation}\label{da26}
    \de_{A_2}=\be_0+\be_1+\be_2=\frac{1}{2}\left(\al_0++\alpha_{1}+\alpha_{5}
    +2 \left( \alpha_{2}+\alpha_{4}+\alpha_{6} \right)
    +3\alpha_{3}\right)=\frac{\de}{2},
\end{equation}
where $\de$ (Equation \eqref{eq:E61nullroot}) is the null root of the $E_6^{(1)}$ root system.  Next, we
look at the actions of $b_i$ $(0\leq i \leq 2)$ on 
the $V_J\bigoplus
V_J^\perp$ basis of
$V^{(1)}$, that is $J\cup \be=\{\al_0, \al_1, \al_3, \al_5, \be_1,
\be_2, \be_0,\}$. These are computed using 
Equation \eqref{A2sr} by
composing the actions of 
$s_j\in W(E_6^{(1)})$ on $V^{(1)}$ given by Equations
\eqref{sij0} and \eqref{CarE6a}. We have,
\begin{subequations}
    \begin{align}
        b_1\colon & (\al_0, \al_1, \al_3, \al_5, \be_1, \be_2, \be_0)\mapsto
        (\al_0, \al_1, \al_5, \al_3, -\be_1, \be_2+\be_1, \be_0+\be_1),
        \\
        b_2\colon &(\al_0, \al_1, \al_3, \al_5, \be_1, \be_2, \be_0)\mapsto
        (\al_0, \al_1, \al_3, \al_5, \be_1+\be_2, -\be_2, \be_0+\be_2),
        \\
        b_0\colon&(\al_0, \al_1, \al_3, \al_5, \be_1+\be_2, \be_2, \be_0)\mapsto
        (\al_0, \al_1, \al_3, \al_5, \be_1+\be_0, \be_2+\be_0, -\be_0).
    \end{align}
    \label{bacte6}
\end{subequations}
From \Cref{bacte6}, we see that $b_i$ $(0\leq i \leq 2)$ acts
on the $\be$-system exactly as the reflection along the root $\be_i$ of
$A_2^{(1)}$ type, while acting permutatively on $J$.  The action of $\sigma$
on the $\be$-system is computed using 
Equations \eqref{bre6} and \eqref{4ddae6}, and we have
\begin{equation}
    \sigma \colon (\be_0,\be_1, \be_2)\mapsto (\be_1,\be_2, \be_0),
    \label{rhor}
\end{equation}
that is, $\sigma$ is an order three diagram automorphism of 
the type $A_2^{(1)}$ $\be$-system.
Hence we see that 
$N_J=\lan b_1, b_2, b_0\ran\rtimes\lan\sigma\ran$
can be realised as 
$\widetilde{W}(A_2^{(1)})$ on the subspace $V_J^\perp=\Span\left(\be\right)$.
\end{proof}

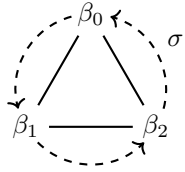
\begin{figure}[hbt]
    \begin{tikzpicture}[thick]
        \node (a1) at (0,1){$\beta_{0}$};
        \node (a2) at ($(-{cos(30)},-{sin(30)})$) {$\beta_{1}$};
        \node (a3) at ($({cos(30)},-{sin(30)})$) {$\beta_{2}$};
        \draw[-] (a1) -- (a2) -- (a3) -- (a1);
        \centerarc[->,dashed](0,0)(-15:75:1);
        \centerarc[->,dashed](0,0)(105:195:1);
        \centerarc[->,dashed](0,0)(225:315:1);
        \node (pi) at ({1.3*cos(30)},{1.3*sin(30)}) {$\sigma$};
    \end{tikzpicture}
    \caption{$\Ga(A_2^{(1)})$ with diagram 
    automorphism $\sigma$. }
    \label{bA2e6}
\end{figure}

Recall that by Equation \eqref{PtoUd},
for an extended affine Weyl group  we have the following
decomposition $\widetilde{W}^{(1)}={W}\ltimes P$. That is,
\begin{equation}\label{ewa2}
N_J=\lan b_1, b_2, b_0\ran\rtimes\lan\sigma\ran
\cong\widetilde{W}\left(A_2^{(1)}\right)
={W}\left(A_2\right)\ltimes P_{A_2}
=\lan b_1, b_2\ran\ltimes \lan U_1,
U_2\ran,    
\end{equation}
where $U_i$ is the translation by the fundamental weight $H_i$ $(i=1,2)$
of this $A_2^{(1)}$ type $\be$-system. In particular we have \cite{Shi:22}:
\begin{equation}
    U_1=\sigma b_2b_1,\quad U_2=\sigma^{-1} b_1b_2.
\end{equation}

$H_i$ $(i=1,2)$ are related to the
finite simple roots of the $\be$-system
by Equations \eqref{pah} and \eqref{hpa},
\beq\label{pahbA2}
\bp
\be_1\\
\be_2
\ep=C(A_2)^T\bp
H_1\\
H_2
\ep=\bp
2&-1\\
-1&2
\ep\bp
H_1\\
H_2
\ep,
\eeq
and
\beq\label{hpabA2}
\bp
H_1\\
H_2
\ep=
\left(C(A_2)^T\right)^{-1}
\bp
\be_1\\
\be_2
\ep=\frac{1}{3}\left(
\begin{array}{cc}
 2 & 1 \\
 1 & 2 \\
\end{array}
\right)\bp
\be_1\\
\be_2
\ep.
\eeq
In particular we have,
\begin{equation}
    \begin{aligned}
        H_2&=\frac{1}{3}(\be_1+2\be_2)
        =\frac{1}{3}\left(\al_4+\frac{\al_{35}}{2}+ 2\left(\al_6+\frac{\al_{03}}{2}\right)\right)
        =\frac{1}{2}\left(h_6-h_2\right),
    \end{aligned}
    \label{H1A2e6}
\end{equation}
where we have used Equation \eqref{pahE6} for the last expression in terms of fundamental
weights $h_i\,(1\leq i \leq 6)$ of $E_6^{(1)}$. 
So we have,
\begin{equation}\label{2H1A2e6}
    2H_2=h_6-h_2.
\end{equation}
Moreover, $|2H_2|^2=|h_6-h_2|^2=\frac{4}{3}$.
Elements $U_2$ and $\left(U_2\right)^{2}$ act on the simple system of $E_6^{(1)}$ by
\begin{equation}
    U_2
    \colon
    \left( 
            \alpha_{0},\alpha_{1},\alpha_{2},
            \alpha_{3},\alpha_{4},\alpha_{5}, \al_6
    \right)
    \mapsto
    \left( 
            \alpha_1, \alpha_0, \alpha_{1223346}, 
            \alpha_5, \alpha_4,
            \al_3, -\alpha_{12345}  
    \right),
    \label{eq:U2}
\end{equation}
and
\begin{equation}
   \left(U_2\right)^{2}
    \colon
    \left( 
        \alpha_{0},\alpha_{1},\alpha_{2},\alpha_{3},
        \alpha_{4},\alpha_{5},\alpha_{6}
    \right)
    \mapsto
    \left( 
        \alpha_{0},\alpha_{1},\alpha_{2}+\de,\alpha_{3},
        \alpha_{4},\alpha_{5},\alpha_{6}-\de
    \right).
    \label{eq:U22}
\end{equation}
By Equation \eqref{uja} we recognise that $U_2^2\in P$ is a translation
by $h_6-h_2$ in $\widetilde{W}(E_6^{(1)})$.
The nature of element $U_2$ becomes more evident in the 
$J\cup \be$ basis
of $V^{(1)}$,
\begin{equation}
    U_2
    \colon
    \left( 
        \begin{gathered}
            \alpha_{0},\alpha_{1},\alpha_{3},\alpha_{5},
            \\
            \be_1, \be_2, \be_0
        \end{gathered}
    \right)
    \mapsto
    \left( 
        \begin{gathered}
            \alpha_{1},\alpha_{0},\alpha_{5},\alpha_{3},
            \\
            \be_1, \be_2-\de_{A_2}, \be_0+\de_{A_2}
        \end{gathered}
\right)
    \label{eq:U2b}.
\end{equation}
Again by Equation \eqref{uja} we see that $U_2$ acts on the $\be$-system
as a translation by $H_2$. However, it also acts on the set $J$
as a permutation of order two $(01)(35)$, hence
can not be an element of translation in $\widetilde{W}(E_6^{(1)})$. 
On the other hand, we have
\begin{equation}
    \left(U_2\right)^{2}
    \colon
    \left( 
        \begin{gathered}
            \alpha_{0},\alpha_{1},\alpha_{3},\alpha_{5},
            \\
            \be_1, \be_2, \be_0
        \end{gathered}
    \right)
    \mapsto
    \left( 
        \begin{gathered}
            \alpha_{0},\alpha_{1},\alpha_{3},\alpha_{5},
            \\
            \be_1, \be_2-2\de_{A_2}, \be_0+2\de_{A_2}
        \end{gathered}
    \right).
    \label{eq:U22b}
\end{equation}
That is, $U_2^2$ acts on the $\be$-system
as a translation by $2H_2=h_6-h_2$ but also fixing $J$.
Hence, $U_2^2$ is a translation by $2H_2=h_6-h_2$ in $\widetilde{W}\left(E_6^{(1)}\right)$ as we knew from \Cref{eq:U22}.

 Let us denote by $u_{ij}$ the element of translation by $h_i-h_j$ in 
    $\widetilde{W}\left(E_6^{(1)}\right)$.
Then we have $u_{ij}=u_iu_j^{-1}$ and
    $(u_{ij})^{-1}=u_{ji}$.  
By \Cref{eq:qrt0E614} we have
$\left(\varphi_{1}^{*}\right)^{4}=u_{46}$,
that is $\varphi_{1}^{*}$ is a quasi-translation of order four for $u_{46}$,
and by \Cref{eq:U22b} $U_2$ is a quasi-translation of order two
for $u_{62}$.
 
A quasi-translation of order four for $u_{62}$ can be constructed simply 
by a rewriting of $U_2$ using the defining relations~\eqref{funWaA2} of 
$\widetilde{W}\left(A_2^{(1)}\right)$,
\begin{equation}
    U_2=\sigma^{-1} b_1b_2
    =\sigma^{2} b_1b_2
    =\sigma b_2\sigma b_2
     =(\sigma b_2)^2.
    \label{U2A2e6}
\end{equation}
Let
\begin{equation}\label{rU2}
    \rho=\sigma b_2=\sigma s_{6036}, \quad\mbox{and we have}\quad \rho^2=U_2,
    \quad \rho^4=u_{62}.
\end{equation}
A conjugation that relates the two order four quasi-translations
$\varphi^{*}_{1}=\sigma b_1s_0s_1$ and $\rho=\sigma b_2$ is found by the 
following observations.

\begin{proposition}\label{u62u46P1}
   $u_{62}=\rho^4$ is related to 
   $u_{46}=\left(\varphi_{1}^{*}\right)^{4}$
   by a conjugation of $b_2\in {W}\left(A_2\right)$,

 \begin{equation}\label{u62u46}
    b_2u_{62}b_2^{-1}=u_{46}.
\end{equation}
\end{proposition}
\begin{proof}
    First, we knew that vectors $h_6-h_2$ and $h_4-h_6$ are of the same length (that is, squared length of $\frac{4}{3}$). In fact, they 
    belong to the 
 same orbit
under the actions of ${W}\left(A_2\right)$,
\begin{equation}
        \begin{aligned}
         b_2(H_2) &= b_2\left(\frac{h_6-h_2}{2}\right) =H_2-(-H_1+2H_2) =
         H_1-H_2
         \\
         &=\frac{1}{3}\left(\be_1-\be_2\right)
        =\frac{1}{2}(h_4-h_6),
        \end{aligned}
        \label{b2H2}
    \end{equation}
    where we have used \Cref{sihjn} and 
    the Cartan matrix \eqref{CarA2a} to compute the action
    of $b_2$ on $H_2$.
    Then by Equation \eqref{ujnj} we have 
    Equation \eqref{u62u46}.
\end{proof}
\begin{proposition}\label{u62u46P2}
    Any conjugation 
    of $b_2\rho b_2^{-1}=\sigma b_1$ by an element of $W_J=\lan s_0, s_1, s_3, s_5\ran$ 
    is a quasi-translation of order four for $u_{46}$. In particular
    we found,
    \begin{equation}\label{conqte6}
  s_1b_2\rho (s_1b_2)^{-1}=\varphi^{*}_{1}.
    \end{equation}
    That is, $\varphi^{*}_{1}$ (given in \Cref{eq:qrt0dec}) is related to 
    the order four quasi-translation $\rho$
    via a conjugation of $s_1b_2$.
\end{proposition}
\begin{proof} By \Cref{u62u46P1} and \Cref{rU2},
\begin{equation}
    u_{46}=b_2u_{62}b_2^{-1}
    =b_2\rho^4 b_2^{-1}
    =\left(b_2\rho b_2^{-1} \right)^4,
\end{equation}
    where we have,
\begin{equation}
   b_2\rho b_2^{-1}= b_2\sigma b_2 b_2^{-1}
    =b_2\sigma
   = \sigma b_1.
\end{equation}
That is, 
\begin{equation}\label{rhou46}
    \left(b_2\rho b_2^{-1} \right)^4=(\sigma b_1)^4=u_{46}.
\end{equation}

Recall from \Cref{NJWJe6} that
\begin{equation}
N(W_J)=N_J\ltimes W_J=\left(\lan b_1, b_2, b_0\ran\rtimes\lan \sigma\ran\right)\ltimes
\lan s_0, s_1, s_3, s_5\ran
\cong \widetilde{W}\left(A_2^{(1)}\right)\ltimes {W}(4A_1).
\end{equation}  
 The root system of $N_J$ $\be$, is orthogonal to
 the root system of $W_J$ $J$.
This means that
any conjugation of $r\in N_J$ by $s\in W_J$, $srs^{-1}$
does not change $r$'s action on the $\be$-system. It only results 
an even number of sign changes (an action of order at most four) in the simple roots of $J$
observed earlier in \Cref{Ne6},
this difference in action is not seen after four iterations.
Equation \eqref{conqte6} is a consequence of this observation and
can be verified by direct computation:
   \begin{equation}\label{conqte6p}
  s_1b_2\rho (s_1b_2)^{-1}=s_1b_2\rho b_2^{-1}s_1^{-1}=s_1\sigma b_1s_1^{-1}=s_1\sigma b_1s_1=\sigma b_1s_0s_1
            =\varphi^{*}_{1},
    \end{equation}
    where we have used the defining relations given in Equation \eqref{funWA2}.
\end{proof}

Hence we have shown that the map $\varphi^{*}_{1}$, as given by 
\Cref{eq:qrt0dec}, is an quasi-translation in a normalizer of $W(4A_1)$ in
$\widetilde{W}(E_6^{(1)})$ (given by \Cref{NWJe6}), and we 
completed the proof of \Cref{thm:mainA}.

\section{Singular fibre $\zeta_{2}(s)$ of type $A_{1}^{(1)}$ associated 
with the map $\varphi_{s}$}
\label{sec:B3}

In this section we compute the action of the map 
$\varphi_{s}$~\eqref{eq:phis} on the orthogonal complement of the 
singular fibre $[0:1]$ of type $A_{1}^{(1)}$ of pencil 
$p_{s}$~\eqref{eq:padd}. Then, we compute the normalizer of the
Weyl group $W(4A_{1})$ in the extended Weyl group 
$\widetilde{W}( E_{7}^{(1)})$. This will yield that the map 
$\varphi_{s}$~\eqref{eq:phis} on such a singular fibre has symmetry type
$B_{3}^{(1)}\subset E_{7}^{(1)}$, a non-simply laced Coxeter group, thus 
ending the proof of \Cref{thm:mainC}.

\subsection{Action on the simple roots of $E_{7}^{(1)}$ and decomposition 
in $\widetilde{W}(E_7^{(1)})$}

Following \Cref{tab:singtetr} we have that the singular fibre of the 
pencil~\eqref{eq:p2} at the point $\left[ \mu:\nu \right]=\left[ 0:1\right]$ 
is of type $A_1^{(1)}$ and it is generated by the following divisors in $\Pic(X_s)$:
\begin{equation}
    D_{0} = H_x+H_y-L_1-L_2-L_3-L_4,
    \quad
    D_{1} = H_x+H_y-L_5-L_6-L_7-L_8.
    \label{eq:phi2Di}
\end{equation}
These two divisors are permuted by the pullback map $\varphi_{s}^{*}$, that is 
$\varphi_{s}^{*}\left( D_{i} \right)=D_{(i+1)}$, $i\in \Z/2\Z$. The rational
elliptic surface $X_s$ with this choice of the singular fibre is shown
in \Cref{fig:picA1s}.
The orthogonal complement of $A_{1}^{(1)}$ in $\Pic\left( X_{s} \right)$
is generated by:
\begin{equation}
    \begin{gathered}
        \alpha_{0} =L_7-L_8,
        \quad
        \alpha_{1} =L_6-L_7,
        \quad
        \alpha_{2} = H_x-H_y,
        \quad
        \alpha_{3} = L_5-L_6,
        \\
        \alpha_{4} = H_y-L_1-L_5,
        \quad
        \alpha_{5} = L_1-L_2,
        \quad
        \alpha_{6} = L_2-L_3,
        \quad
        \alpha_{7} = L_3-L_4.
    \end{gathered}
    \label{eq:E71}
\end{equation}

The set $\De^{(1)}=\De\left(E_7^{(1)}\right)
=\{\al_j\mid 0\leq j\leq 7\}$ is a simple system of type
$E_7^{(1)}$, its Dynkin diagram given in \Cref{fig:e7B}.
The anti-canonical divisor corresponds to the following null-root
in the $E_{7}^{(1)}$ root system:
\begin{equation}
    -K_{X_{2}}=\delta = 
    \alpha_{071122663335554444}.
    \label{eq:E71nullroots}
\end{equation}
Let $V^{(1)}$ be a vector space
with basis 
$\De^{(1)}$,
equipped with a symmetric bilinear form given by Equations \eqref{alaij0}
and \eqref{CarE7a}. The corresponding extended affine Weyl group is,
\begin{equation}
    \widetilde{W}(E_{7}^{(1)})=\langle s_i\;|\;0\leq i \leq
    7\rangle\rtimes\langle \sigma\rangle
    =W(E_{7})\ltimes P=W(E_{7})\ltimes\lan u_j\, \mid 1\leq j\leq 7\ran,
\end{equation}
where $u_j$ is the translation by the fundamental weight 
of the $E_7^{(1)}$ system, $h_j$ $(1\leq j\leq 7)$.
In particular, the action of 
$s_j\in W(E_7^{(1)})$ on $\De^{(1)}$ is given by Equations
\eqref{sij0} and \eqref{CarE7a}.
See \Cref{app:E7} for
basic data and properties of $\widetilde{W}(E_{7}^{(1)})$.

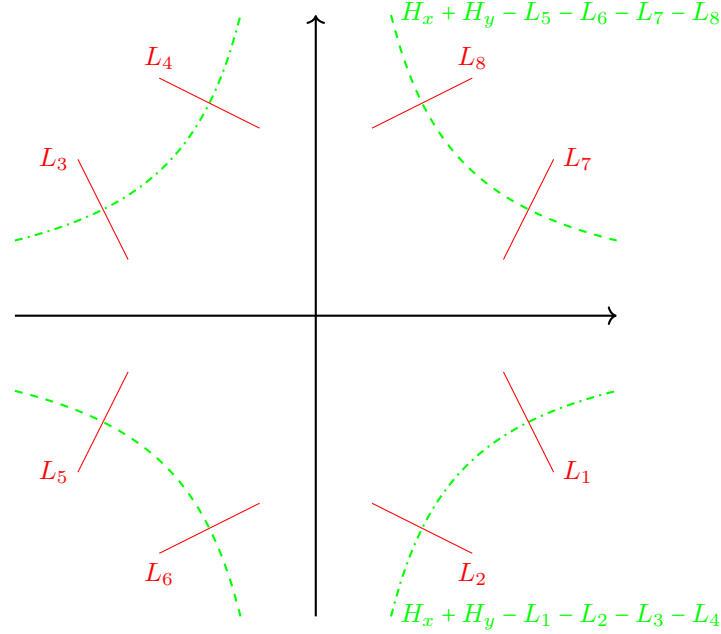
\begin{figure}[ht!]
    \centering
    \begin{tikzpicture}[scale=2]
        \def\pt{2};
        \draw[thick,->] (0,-2)--(0,2);
        \draw[thick,->] (-2,0)--(2,0);
        \draw[domain=1/2:2, smooth, variable=\x, green,thick,dashed] plot (\x, 1/\x);
        \draw[domain=-2:-1/2, smooth, variable=\x, green,thick,dashed] plot (\x, 1/\x);
        \draw[domain=1/2:2, smooth, variable=\x, green,thick,dash dot] plot (\x, -1/\x);
        \draw[domain=-2:-1/2, smooth, variable=\x, green,thick,dash dot] plot (\x, -1/\x);
        \node[green,above,right] at (1/2,2) {$H_{x}+H_{y}-L_{5}-L_{6}-L_{7}-L_{8}$};
        \node[green,below,right] at (1/2,-2) {$H_{x}+H_{y}-L_{1}-L_{2}-L_{3}-L_{4}$};
        \node[red,right] at ($({sqrt(\pt)+1/6},{-1/sqrt(\pt)-1/3})$) (F1) {$L_{1}$};
        \draw[red] ($({sqrt(\pt)+1/6},{-1/sqrt(\pt)-1/3})$) --($(-{-sqrt(\pt)-1/6},{-1/sqrt(\pt)+1/3})$) ;
        \node[red,below] at ($({1/sqrt(\pt)+1/3},{-sqrt(\pt)-1/6})$) (F2) {$L_{2}$};
        \draw[red] ($({1/sqrt(\pt)+1/3},{-sqrt(\pt)-1/6})$) --($({1/sqrt(\pt)-1/3},{-sqrt(\pt)+1/6})$) ;
        \node[red,left] at ($({-sqrt(\pt)-1/6},{1/sqrt(\pt)+1/3})$) (F3) {$L_{3}$};
        \draw[red] ($({-sqrt(\pt)-1/6},{1/sqrt(\pt)+1/3})$) --($(-{sqrt(\pt)+1/6},{1/sqrt(\pt)-1/3})$) ;
        \node[red,above] at ($({-1/sqrt(\pt)-1/3},{sqrt(\pt)+1/6})$) (F6) {$L_{4}$};
        \draw[red] ($({-1/sqrt(\pt)-1/3},{sqrt(\pt)+1/6})$) --($({-1/sqrt(\pt)+1/3},{sqrt(\pt)-1/6})$) ;
        \node[red,left] at ($({-sqrt(\pt)-1/6},{-1/sqrt(\pt)-1/3})$) (F5) {$L_{5}$};
        \draw[red] ($({-sqrt(\pt)-1/6},{-1/sqrt(\pt)-1/3})$) --($(-{sqrt(\pt)+1/6},{-1/sqrt(\pt)+1/3})$) ;
        \node[red,below] at ($({-1/sqrt(\pt)-1/3},{-sqrt(\pt)-1/6})$) (F6) {$L_{6}$};
        \draw[red] ($({-1/sqrt(\pt)-1/3},{-sqrt(\pt)-1/6})$) --($({-1/sqrt(\pt)+1/3},{-sqrt(\pt)+1/6})$) ;
        \node[red,right] at ($({sqrt(\pt)+1/6},{1/sqrt(\pt)+1/3})$) (E1) {$L_{7}$};
        \draw[red] ($({sqrt(\pt)-1/6},{1/sqrt(\pt)-1/3})$)--($({sqrt(\pt)+1/6},{1/sqrt(\pt)+1/3})$) ;
        \node[red,above] at ($({1/sqrt(\pt)+1/3},{sqrt(\pt)+1/6})$) (E4) {$L_{8}$};
        \draw[red] ($({1/sqrt(\pt)-1/3},{sqrt(\pt)-1/6})$) --($({1/sqrt(\pt)+1/3},{sqrt(\pt)+1/6})$) ;
    \end{tikzpicture}
    \caption{The space of initial values of \eqref{eq:phis} interpreted
        as elliptic surface of type $A_{1}^{(1)}$.}
    \label{fig:picA1s}
\end{figure}

\def\centerarc[#1](#2)(#3:#4:#5);%
    {
    \draw[#1]([shift=(#3:#5)]#2) arc (#3:#4:#5);
    }
\begin{figure}[ht!]
    \centering
    \begin{tikzpicture}[thick]
        \node (a1) at (0,0){$\alpha_{0}$};
        \node (a2) at (1,0) {$\alpha_{1}$};
        \node (a3) at (2,0) {$\alpha_{3}$};
        \node (a4) at (3,0) {$\alpha_{4}$};
        \node (a5) at (4,0) {$\alpha_{5}$};
        \node (a6) at (5,0) {$\alpha_{6}$};
        \node (a7) at (6,0) {$\alpha_{7}$};
        \node (a0) at (3,1) {$\alpha_{2}$};
        \draw[-] (a1) -- (a2) -- (a3)-- (a4)-- (a5) --(a6) --(a7);
        \draw[-] (a0) -- (a4);
        \centerarc[<->,dashed](3,0)(-5:-175:1);
        \centerarc[<->,dashed](3,0)(-4:-176:2);
        \centerarc[<->,dashed](3,0)(-3:-177:3);
        \node (pi) at (3,-3.3) {$\sigma$};
    \end{tikzpicture}
    \caption{The Dynkin diagram of the Weyl group $\widetilde{W}( E_{7}^{(1)} )$ in
    Bourbaki's numbering.}
    \label{fig:e7B}
\end{figure}
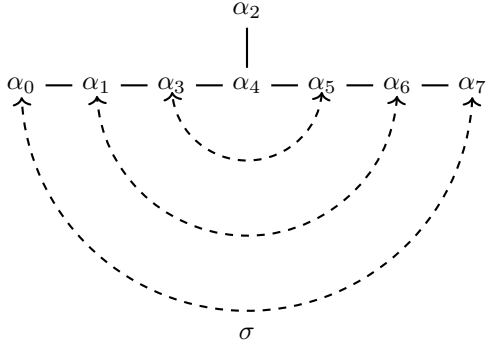
Using Equations \eqref{eq:phi3pb} and \eqref{eq:E71}
we have $\varphi^{*}_{s}$ acts on $\De^{(1)}$ by:
\begin{equation}
    \varphi^{*}_{s}
    \colon
    \left(
        \begin{gathered}
            \alpha_{0},\alpha_{1},\alpha_{2},\alpha_{3},
            \\
            \alpha_{4},\alpha_{5},\alpha_{6},\alpha_{7}
        \end{gathered}
    \right)
    \mapsto
    \left( 
        \begin{gathered}
            \alpha_3, -\alpha_{234567}, \alpha_{23445}, \alpha_7,
            \\
            \alpha_6, \alpha_5, \alpha_{1234}, \alpha_0
        \end{gathered}
\right).
    \label{eq:phisE71s}
\end{equation}
Compute three iterations of the map $\varphi^{*}_s$ we obtain:
\begin{equation}
    \left(\varphi_{s}^{*}\right)^3
    \colon
    \left(
        \begin{gathered}
            \alpha_{0},\alpha_{1},\alpha_{2},\alpha_{3},
            \\
            \alpha_{4},\alpha_{5},\alpha_{6},\alpha_{7}
        \end{gathered}
    \right)
    \mapsto
    \left(
        \begin{gathered}
            \alpha_{0},\alpha_{1}-\delta,\alpha_{2}+\delta,\alpha_{3},
            \\
            \alpha_{4},\alpha_{5},\alpha_{6},\alpha_{7}
        \end{gathered}
    \right).
    \label{eq:phis3E71s}
\end{equation}
By Equation \eqref{uja} we recognise that
$\left(\varphi_{s}^{*}\right)^{3}\in P$ is
the translation
by $h_1-h_2$,
and $\varphi_{s}^{*}$, a quasi-translation of order three in $\widetilde{W}\left(E_7^{(1)}\right)$.
Moreover, using \Cref{hijE7} one can check that 
$|h_1-h_2|^2 =\frac{3}{2}$, hence
it is in the same $W(E_7)$-orbit of $h_7$.
Finally, $\varphi_{s}^{*}$ has the following decomposition
in terms of the generators of $\widetilde{W}( E_{7}^{(1)} )$ using Equations \eqref{lfun} and \eqref{uaw}:
\begin{equation}
    \varphi^{*}_{s} =
    \sigma s_2 s_4 s_5 s_3 s_4 s_1 s_3 s_0 s_1.
    \label{eq:phisdec}
\end{equation}

\subsection{Normalizer of ${W}\bigl(4A_1\bigr)$ in $\widetilde{W}\left(E_7^{(1)}\right)$.}
\label{Ne7s}
The initial observation on $\varphi^{*}_{s}$ comes from its actions on $\De^{(1)}$
given in Equation \eqref{eq:phisE71s}, that is, $\varphi^{*}_{s}$ permutes the subset  
$\{\al_0, \al_3, \al_5, \al_7\}\subset \De^{(1)}$. This prompts us to take
\begin{equation}\label{Je7}
J=\{\al_0, \al_3, \al_5,\al_7\}\cong 4A_1, 
\end{equation}
and look at the normalizer of 
$W_J=\lan s_0, s_3, s_5, s_7\ran$ in $\widetilde{W}(E_7^{(1)})$.
\begin{proposition}\label{NJWJe7s}
  The normalizer of ${W}_J$ with $J=\{\al_0, \al_3, \al_5,\al_7\}$ in $\widetilde{W}( E_{7}^{(1)} )$ is given by
\begin{equation}\label{NWJe7}
\begin{aligned}
 N(W_J)=N_J\ltimes W_J&=\left(\lan b_1, b_2, b_3, b_0\ran\rtimes\lan \sigma\ran\right)\ltimes
\lan s_0, s_3, s_5, s_7\ran\\
&\cong \widetilde{W}\left(B_3^{(1)}\right)\ltimes {W}(4A_1),   
\end{aligned}
\end{equation}  
with
\begin{equation}\label{B3sr}
    b_1 = s_{6576},\quad
    b_2 = s_{4534}, \quad
    b_3 = s_{2}, \quad
    b_0 = s_{1301}.
\end{equation}
\end{proposition}
\begin{proof}
 Recall that $N(W_J)=N_J\ltimes W_J$, 
where $N_J$ generated by the R- and M-elements,
is the group of all elements of $\widetilde{W}(E_{7}^{(1)})$
that act permutatively on the subset $J\subset \De^{(1)}$. Here we have 
$J=\{\al_0, \al_3, \al_5,\al_7\}\cong 4A_1$, and
$W_J=\lan s_0, s_3, s_5, s_7\ran\cong {W}(4A_1)$. 
Following the procedure of computing normalizers given in
\cite{H, BH}, we find that the
R-elements are given by,
\begin{equation}
    \label{B3srp}
    b_1 = s_{6576},\quad
    b_2 = s_{4534}, \quad
    b_3 = s_{2}, \quad
    b_0 = s_{1301}.
\end{equation}
where $b_1, b_2, b_0$ act on the subset $J$ as permutations of order two,
written
as a permutation on the index set of $J$ we have,
\begin{equation}
    b_1=(57),
    \quad
    b_2=(35),
    \quad
    b_0=(03).
\end{equation}
The M-element is given by $\sigma$ 
(its actions on $\De^{(1)}$ are given in Equation \eqref{ddae7}), 
its action on $J$ is,
\begin{equation}\label{4ddae7J}
    \sigma=(35)(07).
\end{equation}
It can be verified using the defining relations of $\widetilde{W}(E_7^{(1)})$
that $b_i$'s are involutions that satisfy the defining relations
\eqref{funWB3} for an affine Weyl group of type $B_3^{(1)}$, while
$\sigma$ is the diagram 
    automorphism of  $\Ga(B_3^{(1)})$ (see Figure \ref{rsB3ae}), and it acts on the $b_i$'s by Equation \eqref{funWaB3},
\begin{subequations}
    \begin{align}
        \label{funWB3}
        &b_j^2=1, \;(j\in \{0,1,2,3\}),\quad(b_1b_{2})^3=(b_1b_{3})^2=(b_2b_{3})^4=1,\\\nonumber
        &(b_0b_{2})^3=(b_0b_{3})^2=(b_{0}b_1)^2=1\\\label{funWaB3}
         &\sigma^2=1, \;\;\sigma b_0=b_1\sigma.
    \end{align}
\end{subequations}
That is, together the R- and M-elements generate $N_J$,
an extended affine Weyl group of type $B_3^{(1)}$,
\begin{equation}\label{nae7}
    N_J=\lan b_1, b_2, b_3, b_0\ran\rtimes\lan \sigma\ran\cong \widetilde{W}(B_3^{(1)}).
\end{equation} 
This ends the proof.
\end{proof}

\begin{figure}[ht]
    \centering
    \begin{tikzpicture}[scale=1.5]
        \node (a0) at ({cos(120)},{sin(120)}){$\beta_{0}$};
        \node (a1) at ({cos(120)},{-sin(120)}){$\beta_{1}$};
        \node (a2) at (0,0){$\beta_{2}$};
        \node (a3) at (1,0) {$\beta_{3}$};
        \draw[double distance=1.5pt] (a2) -- (a3);
        \node (b) at ($(a2)!0.5!(a3)$) {$>$};
        \draw[-,thick] (a0) -- (a2);
        \draw[-,thick] (a2) -- (a1);
        \centerarc[<->,dashed,thick](0,0)(130:230:1);
        \node (pi) at (-1.2,0) {$\sigma$};
    \end{tikzpicture}
    \caption{Dynkin diagram of affine $B_3$ type with the diagram automorphism, $\widetilde\Ga(B_3^{(1)})$.}\label{rsB3ae}
\end{figure}
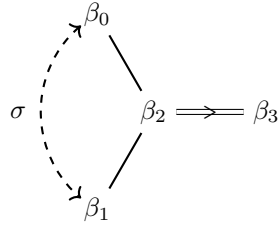
\subsubsection{Sub-root system and translations in
$N_J\cong\widetilde{W}( B_3^{(1)} )$}
\label{Te7}

First we find a subspace $V_J^\perp\subset V^{(1)}$ on which $b_i\in N_J$
$(0\leq i \leq 3)$ can be realised as simple reflections of $B_3^{(1)}$
type. We construct an element of translation in $\widetilde{W}( B_3^{(1)}
)$ and find its relation to the element
$\varphi^{*}_{s}$ given in \Cref{eq:phisdec}.

\begin{proposition}\label{NJWJe7rs}
    The root system for $N_J$ given in \Cref{nae7} generated by 
    $\be=\{\be_i\mid 0\leq i \leq 3\}$ with
    \begin{align}\label{bre7}
        \be_0&=\al_{0113},\quad
        \be_1=\al_{5667},\quad
        \be_2=\al_{3445},\quad
        \be_3=\al_{2},\quad
    \end{align} is of type $B_3^{(1)}$, 
    and we have $V_J^\perp=\Span\left(\be\right)$. 
    The group $N_J=\lan b_1, b_2, b_3, b_0\ran\rtimes\lan\sigma\ran$ can be realised as 
    $\widetilde{W}(B_3^{(1)})$ on $V_J^\perp$ with $b_i$ acting as the 
    reflection along the root $\be_i$, while $\sigma$ exchanges
    $\be_0$ and $\be_1$, which is the diagram 
    automorphism of  $\Ga(B_3^{(1)})$ for the $\be$-system (see Figure \ref{rsB3ae}).
\end{proposition}

\begin{proof}
    Let $\be=\{\be_i\mid 0\leq i \leq 3\}$ 
    be as given in Equation \eqref{bre7}.
    The bilinear form on $V^{(1)}$ in  $\De^{(1)}$ basis given by Equations \eqref{alaij0} and \eqref{CarE7a} is used to check that the $\be$-system
    is of dimension four and orthogonal to $V_J$. That is, $V_J^\perp=\Span\left(\be\right)$.
    Moreover, we found that $|\be_i|^2=4$,
    so $\oc\be_j=\frac{\be_j}{2}$ (for $0\leq i \leq 2$),
    and $|\be_3|^2=2$ so $\oc\be_3=\be_3$.
    Computing the values of $\be_i\cdot\oc\be_j$
    for the $\be$-system \eqref{bre7} using the bilinear form on $V^{(1)}$ we have,
    \begin{equation}\label{CarB3a}
     C(B_3^{(1)})=(\be_i\cdot\oc\be_j)_{1\leq i,j\leq 3, 0}=\left(
    \begin{array}{cccc}
     2 & -1 & 0 & -1 \\
     -1 & 2 & -2 & 0 \\
     0 & -1 & 2 & 0 \\
     -1 & 0 & 0 & 2 \\
    \end{array}
    \right).
    \end{equation}
    The matrix in Equation \eqref{CarB3a}
    is the Cartan matrix of type $B_3^{(1)}$.
    That is, $\be$-system is of $B_3^{(1)}$ type, its
    Dynkin diagram given in Figure \ref{rsB3ae}.
    The null root of this $B_3^{(1)}$ system given by
    \begin{equation}\label{db3}
        \de_{B_3}=\be_0+\be_1+2\be_2+2\be_3
        =\alpha_{071122663335554444}=\de,
    \end{equation}
    where $\de$ (Equation \eqref{eq:E71nullroots}) is the null root of the $E_7^{(1)}$ root system.
    The action of $\sigma$
    on the $\be$-system is computed using 
    Equations \eqref{bre7} and \eqref{ddae7}, and we have
    \begin{equation}
        \sigma \colon (\be_0,\be_1, \be_2,\be_3)\mapsto 
        (\be_1, \be_0, \be_2,\be_3),
        \label{rhor7}
    \end{equation}
    that is, $\sigma$ is an order two diagram automorphism of 
    this type $B_3^{(1)}$ $\be$-system.
    Hence we see that 
    $N_J=\lan b_1, b_2, b_3, b_0\ran\rtimes\lan\sigma\ran$
    can be realised as 
    $\widetilde{W}(B_3^{(1)})$ on the subspace $V_J^\perp=\Span\left(\be\right)$.
\end{proof}

Recall that by \Cref{PtoUd}, for an extended affine Weyl group  we have 
the following decomposition $\widetilde{W}^{(1)}={W}\ltimes P$. That is,
\begin{equation}\label{ewb3}
N_J=\lan b_1, b_2, b_3, b_0\ran\rtimes\lan\sigma\ran
\cong\widetilde{W}\left(B_3^{(1)}\right)
={W}\left(B_3\right)\ltimes P_{B_3}
=\lan b_1, b_2, b_3\ran\ltimes \lan U_1,
U_2, U_3\ran,    
\end{equation}
where $U_i$ is the translation by the fundamental weight $H_i$ $(i=1, 2, 3)$
of the type $B_3^{(1)}$ $\be$-system. In particular we have \cite{Shi:22}:
\begin{equation}\label{Ub3}
    U_1=\sigma b_{12321},
    \quad U_2=b_{02321232},
    \quad U_3=\sigma b_{123023123}.
\end{equation}
$H_i$ $(i=1, 2, 3)$ are related to the
finite simple roots of the $\be$-system
by Equations \eqref{pah} and \eqref{hpa},
\beq\label{pahbC3}
\bp
\pi(\oc\be_1)\\
\pi(\oc\be_2)\\
\pi(\oc\be_3)
\ep=C(B_3)^T\bp
H_1\\
H_2\\
H_3
\ep=\left(
\begin{array}{ccc}
 2 & -1 & 0 \\
 -1 & 2 & -1 \\
 0 & -2 & 2 \\
\end{array}
\right)\bp
H_1\\
H_2\\
H_3
\ep,
\eeq
and
\beq\label{hpabC3}
\bp
H_1\\
H_2\\
H_3
\ep=
\left(C(B_3)^T\right)^{-1}
\bp
\pi(\oc\be_1)\\
\pi(\oc\be_2)\\
\pi(\oc\be_3)
\ep=\bp
1&1&\frac{1}{2}\\
1&2&1\\
1&2&\frac{3}{2}
\ep\bp
\pi(\oc\be_1)\\
\pi(\oc\be_2)\\
\pi(\oc\be_3)
\ep.
\eeq
Recall that for non-simply-laced systems one 
can identify $\pi(\oc\be_i)$ with $\frac{2\be_i}{|\be_i|^2}$,
so we can express coroots of the $B_3^{(1)}$ system
in terms of the $E_7^{(1)}$ simple roots,
\begin{equation}
    \label{pbal}
    \begin{gathered}
        \pi(\oc\be_0)=\frac{\al_{03}}{2}+\al_1,\quad
        \pi(\oc \be_1)=\frac{\al_{57}}{2}+\al_6,\quad
        \pi(\oc \be_2)=\frac{\al_{35}}{2}+\al_4,\quad
        \pi(\oc  \be_3)=\al_2. 
    \end{gathered}
\end{equation}
In particular, we have
\begin{equation}
\begin{aligned}
 H_3&=\pi(\oc \be_1)+2 \pi(\oc \be_2)+\frac{3}{2}\pi(\oc  \be_3),\\
    &= \frac{3}{2}\al_2+\al_3+2\al_4+\frac{3}{2}\al_5+\al_6+\frac{1}{2}\al_7,\\
    &=-h_1+h_2.
\end{aligned}
\end{equation}
where we have used Equation \eqref{pahE7} for the last expression in terms of fundamental
weights $h_i\,(1\leq i \leq 7)$ of $E_7^{(1)}$.
That is, $U_3\in P$ defined in \Cref{Ub3}
is a translation by $-h_1+h_2$ in $\widetilde{W}\left(E_7^{(1)}\right)$.

\begin{proposition}
    The map $\varphi_{s}^{*}$ given in \Cref{eq:phisdec} is an 
    element of quasi-translation of order three in the  
    $N_J\cong\widetilde{W}(B_3^{(1)})$ subgroup of 
    $\widetilde{W}(E_7^{(1)})$ given in \Cref{NJWJe7s}. We have,
    \begin{equation}\label{phisB3}
        \varphi_{s}^{*}=\sigma s_2 s_{4534} s_{1301}=\sigma b_{320}, 
    \end{equation} 
    where $\left(\varphi_{s}^{*}\right)^{3}=U_3^{-1}$ with $U_3$ given 
    by \Cref{Ub3}.
    \label{prop:phisact}
\end{proposition}

\begin{proof}
    The expression given in \Cref{phisB3} in terms of the generators
    of $N_J$ can be written down directly by inspection on 
    \Cref{eq:phisdec,B3sr}. To show that $\varphi_{s}^{*}$ is an element 
    of quasi-translation of $N_J$, recall that by \Cref{eq:phis3E71s}
    $\left(\varphi_{s}^{*}\right)^{3}\in P$ is the translation by 
    $h_1-h_2$. That is, $\varphi_{s}^{*}$ is a quasi-translation of order 
    three to the element $U_3^{-1}$. A quasi-translation of order three
    for $U_3^{-1}$ can be constructed simply by a rewriting
    of $U_3^{-1}$ using the defining relations \eqref{funWaB3} of 
    $\widetilde{W}\left(B_3^{(1)}\right)$,
    \begin{equation}\label{U2B3e7}
    \begin{aligned}
    U_3^{-1}&=b_{321320321}\sigma,\\
    &=\sigma b_{320}b_{321}b_{320},\\
    &=\sigma b_{320}b_{321}\sigma^2b_{320},\\
    &=\sigma b_{320}\sigma b_{320}\sigma b_{320}.
    \end{aligned}
    \end{equation}
    In fact, we have
    \begin{equation}\label{phisB3p}
    \varphi_{s}^{*}=\sigma s_2 s_{4534} s_{1301}=\sigma b_{320},\quad
    \left(\varphi_{s}^{*}\right)^{3}=U_3^{-1}.
    \end{equation}
    The action of $\varphi_{s}^{*}$ on the $J\cup \be=\{\al_0, \al_3, \al_5,\al_7,\be_1, \be_2,\be_3,\be_0\}$ basis is given by,
    \begin{equation}
        \varphi_{s}^{*}
        \colon
        \left(
            \begin{gathered}
                \al_0, \al_3, \al_5,\al_7,
                \\
                \be_1, \be_2,\be_3,\be_0
            \end{gathered}
        \right)
        \mapsto
        \left( 
            \begin{gathered}
                \al_3, \al_7, \al_5,\al_0, 
                \\
                \be_{2330}, \be_1,\be_{23},-\be_{1233}  
            \end{gathered}
    \right),
        \label{eq:phisu}
    \end{equation}
    moreover we have,
    \begin{equation}
        \left(\varphi_{s}^{*}\right)^3
        \colon
        \left(
            \begin{gathered}
                \al_0, \al_3, \al_5,\al_7,
                \\
                \be_1, \be_2,\be_3,\be_0
            \end{gathered}
        \right)
        \mapsto
        \left( 
            \begin{gathered}
                \al_0, \al_3, \al_5,\al_7, 
                \\
               \be_1, \be_2, \be_3+\de,\be_0-2\de
            \end{gathered}
        \right),
        \label{eq:phis3u}
    \end{equation}
 thus ending the proof.
\end{proof}

Hence, the proof of last statement of \Cref{thm:mainC} follows
from \Cref{prop:phisact}.

\section{Singular fibre of type $A_{1}^{(1)}$ associated with the map $\varphi_{2}$}
\label{sec:C3}

In this section we compute the action of the map 
$\varphi_{2}$~\eqref{eq:phi3} on the orthogonal complement of the 
singular fibre $[0:1]$ of type $A_{2}^{(1)}$ of pencil 
$p_{2}$~\eqref{eq:p1}. Then, we compute the normalizer of the
Weyl group $W(4A_{1})$ in the extended Weyl group 
$\widetilde{W}( E_{6}^{(1)})$. In this case, this will not be
directly as in the \Cref{sec:B3}, but we will use the results
of that section to show that the mentioned action is on a 
translated root lattice of type $B_{3}^{(1)}$. So, the final
outcome of this section is that the map $\varphi_{2}$~\eqref{eq:phi3} on 
the singular fibre $[0:1]$ has symmetry type 
$B_{3}^{(1)}\subset E_{7}^{(1)}$, thus ending the proof of 
\Cref{thm:mainB}.

\subsection{Action on the simple roots of $E_{7}^{(1)}$ and decomposition in $\widetilde{W}(E_7^{(1)})$}

Following \Cref{tab:singtetr} we have that the singular fiber of the 
pencil \eqref{eq:p2} at the point $\left[ \mu:\nu \right]=\left[ 0:1\right]$ 
is of type $A_1^{(1)}$ and it is generated by the following divisors in $\Pic(X_2)$:
\begin{equation}
    D_{0} = H_x+H_y-F_1-F_2-F_3-F_4,
    \quad
    D_{1} = H_x+H_y-F_5-F_6-F_7-F_8.
    \label{eq:phi3Di}
\end{equation}
These two divisors are permuted by the pullback map $\varphi_{2}^{*}$, that is 
$\varphi_{2}^{*}\left( D_{i} \right)=D_{(i+1)}$, $i\in \Z/2\Z$. The rational
elliptic surface $X_2$ with this choice of the singular fibre is shown
in \Cref{fig:picA12}.
The orthogonal complement of $A_{1}^{(1)}$ in $\Pic\left( X_{2} \right)$
is generated by the same formula as \Cref{eq:E71} with 
$F_{i}\longleftrightarrow L_{i}$.

The set $\De^{(1)}=\De(E_7^{(1)}) =\{\al_j\mid 0\leq j\leq 7\}$ is a 
simple system of type $E_7^{(1)}$, its Dynkin diagram given in 
\Cref{fig:e7B}. The anti-canonical divisor corresponds to the following 
null-root in the $E_{7}^{(1)}$ root system:
\begin{equation}
    -K_{X_{2}}=\delta = 
    \alpha_{071122663335554444}.
    \label{eq:E71nullroot2}
\end{equation}
Let $V^{(1)}$ be a vector space with basis $\De^{(1)}$, equipped with a
symmetric bilinear form given by \Cref{alaij0,CarE7a}. The corresponding 
extended affine Weyl group is,
\begin{equation}
    \widetilde{W}(E_{7}^{(1)})=\langle s_i\;|\;0\leq i \leq
    7\rangle\rtimes\langle \sigma\rangle
    =W(E_{7})\ltimes P=W(E_{7})\ltimes\lan u_j\, \mid 1\leq j\leq 7\ran,
\end{equation}
where $u_j$ is the translation by the fundamental weight 
of the $E_7^{(1)}$ system, $h_j$ $(1\leq j\leq 7)$.
In particular, the action of 
$s_j\in W(E_7^{(1)})$ on $\De^{(1)}$ is given by Equations
\eqref{sij0} and \eqref{CarE7a}.
See \Cref{app:E7} for
basic data and properties of $\widetilde{W}(E_{7}^{(1)})$.

\begin{figure}[ht!]
    \centering
    \begin{tikzpicture}[scale=2]
        \def\pt{2};
        \draw[thick,->] (0,-2)--(0,2);
        \draw[thick,->] (-2,0)--(2,0);
        \draw[domain=1/2:2, smooth, variable=\x, green,thick,dashed] plot (\x, 1/\x);
        \draw[domain=-2:-1/2, smooth, variable=\x, green,thick,dashed] plot (\x, 1/\x);
        \draw[domain=1/2:2, smooth, variable=\x, green,thick,dash dot] plot (\x, -1/\x);
        \draw[domain=-2:-1/2, smooth, variable=\x, green,thick,dash dot] plot (\x, -1/\x);
        \node[green,above,right] at (1/2,2) {$H_{x}+H_{y}-F_{5}-F_{6}-F_{7}-F_{8}$};
        \node[green,below,right] at (1/2,-2) {$H_{x}+H_{y}-F_{1}-F_{2}-F_{3}-F_{4}$};
        \node[red,right] at ($({sqrt(\pt)+1/6},{-1/sqrt(\pt)-1/3})$) (F1) {$F_{1}$};
        \draw[red] ($({sqrt(\pt)+1/6},{-1/sqrt(\pt)-1/3})$) --($(-{-sqrt(\pt)-1/6},{-1/sqrt(\pt)+1/3})$) ;
        \node[red,below] at ($({1/sqrt(\pt)+1/3},{-sqrt(\pt)-1/6})$) (F2) {$F_{2}$};
        \draw[red] ($({1/sqrt(\pt)+1/3},{-sqrt(\pt)-1/6})$) --($({1/sqrt(\pt)-1/3},{-sqrt(\pt)+1/6})$) ;
        \node[red,left] at ($({-sqrt(\pt)-1/6},{1/sqrt(\pt)+1/3})$) (F3) {$F_{3}$};
        \draw[red] ($({-sqrt(\pt)-1/6},{1/sqrt(\pt)+1/3})$) --($(-{sqrt(\pt)+1/6},{1/sqrt(\pt)-1/3})$) ;
        \node[red,above] at ($({-1/sqrt(\pt)-1/3},{sqrt(\pt)+1/6})$) (F6) {$F_{4}$};
        \draw[red] ($({-1/sqrt(\pt)-1/3},{sqrt(\pt)+1/6})$) --($({-1/sqrt(\pt)+1/3},{sqrt(\pt)-1/6})$) ;
        \node[red,left] at ($({-sqrt(\pt)-1/6},{-1/sqrt(\pt)-1/3})$) (F5) {$F_{5}$};
        \draw[red] ($({-sqrt(\pt)-1/6},{-1/sqrt(\pt)-1/3})$) --($(-{sqrt(\pt)+1/6},{-1/sqrt(\pt)+1/3})$) ;
        \node[red,below] at ($({-1/sqrt(\pt)-1/3},{-sqrt(\pt)-1/6})$) (F6) {$F_{6}$};
        \draw[red] ($({-1/sqrt(\pt)-1/3},{-sqrt(\pt)-1/6})$) --($({-1/sqrt(\pt)+1/3},{-sqrt(\pt)+1/6})$) ;
        \node[red,right] at ($({1+1/2},{1+1/2})$) (F7) {$F_{7}-F_8$};
        \draw[red] ($({1-1/6},{1-1/6})$) --($({1+1/2},{1+1/2})$) ;
        \node[red,right] at ($({4/3+1/3},{4/3-1/3})$) (F7) {$F_8$};
        \draw[red] ($({4/3-1/6},{4/3+1/6})$) --($({4/3+1/3},{4/3-1/3})$) ;
    \end{tikzpicture}
    \caption{The space of initial values of \eqref{eq:phi3} interpreted
        as elliptic surface of type $A_{1}^{(1)}$.}
    \label{fig:picA12}
\end{figure}

    Using Equations \eqref{eq:phi3pb} and \eqref{eq:E71}
we have $\varphi^{*}_{s}$ acts on $\De^{(1)}$ by:   
\begin{equation}
    \varphi^{*}_{2}
    \colon
    \left(
        \begin{gathered}
            \alpha_{0},\alpha_{1},\alpha_{2},\alpha_{3},
            \\
            \alpha_{4},\alpha_{5},\alpha_{6},\alpha_{7}
        \end{gathered}
    \right)
    \mapsto
    \left( 
        \begin{gathered}
            \alpha_{0}, -\alpha_{123450}, 
            \alpha_{11233445670},-\alpha_6, 
            \\
            \alpha_{2456},-\alpha_{24},
            -\alpha_{567}, 
            \alpha_{234567}
        \end{gathered}
\right).
    \label{eq:phi3E71B}
\end{equation}
Computing three iterations of the map $\varphi^{*}_2$ we obtain:
\begin{equation}
    \left(\varphi_{2}^{*}\right)^3
    \colon
    \left(
        \begin{gathered}
            \alpha_{0},\alpha_{1},\alpha_{2},\alpha_{3},
            \\
            \alpha_{4},\alpha_{5},\alpha_{6},\alpha_{7}
        \end{gathered}
    \right)
    \mapsto
    \left(
        \begin{gathered}
            \alpha_{0},\alpha_{1},\alpha_{2}+\de,\alpha_{3},
            \\
            \alpha_{4},\alpha_{5}-\de,\alpha_{6},\alpha_{7}+\de
        \end{gathered}
    \right).
    \label{eq:phi33E71B}
\end{equation}
By Equation \eqref{uja}, we recognise that
$\left(\varphi_{2}^{*}\right)^{3}\in P$ is
the translation
by $-h_2+h_5-h_7$,
and $\varphi_{2}^{*}$, a quasi-translation of order three in $\widetilde{W}\left(E_7^{(1)}\right)$.
Using \Cref{hijE7}, we find that
$|-h_2+h_5-h_7|^2=\frac{3}{2}$.
Finally, $\varphi_{2}^{*}$ has the following decomposition
in terms of the generators of $\widetilde{W}( E_{7}^{(1)} )$ using Equations \eqref{lfun} and \eqref{uaw}:
\begin{equation}
    \varphi_{2}^{*} = \sigma s_{6524310765431463456}.
    \label{eq:phi3decB}
\end{equation}

\subsection{Normalizer of ${W}\bigl(4A_1\bigr)$ in $\widetilde{W}\left(E_7^{(1)}\right)$.}
\label{Ne72}
The initial observation on $\varphi^{*}_{2}$ given in \Cref{eq:phi3decB},
is that like the $\varphi^{*}_{s}$ discussed in \Cref{sec:B3} (where $\varphi^{*}_{s}$ was shown to be an element
 of the normalizer for a standard
parabolic subgroup $W_J$ with $J=\{\al_0, \al_3, \al_5,\al_7\}\cong 4A_1$),
$\varphi^{*}_{2}$ is an order three quasi-translation by a vector of squared length $\frac{3}{2}$.
However, unlike $\varphi^{*}_{s}$ 
the actions of $\varphi_{2}^{*}$ on $\De^{(1)}$
given in \Cref{eq:phi3E71B} shows no clear permutations on any subset of $\De^{(1)}$.
Here we show that $\varphi_{2}^{*}$ in fact is an element of the normalizer for another
type $4A_1$ parabolic subgroup $W_{uJ}$ for some 
$u\in \widetilde{W}\left(E_7^{(1)}\right)$.
To find the appropriate $u$ in ${W}\left(E_7^{(1)}\right)$ we use the
fact that the two squared length $\frac{3}{2}$ vectors associate to
$\varphi_{2}^{*}$ and $\varphi_{s}^{*}$ are related by,
\begin{equation}
    u^{-1}(h_1-h_2)=-h_2+h_5-h_7,\quad\mbox{where}\quad
    u=s_{5436542}\in \widetilde{W}\left(E_7^{(1)}\right).
\end{equation}
By the contragredient action of $\widetilde{W}\left(E_7^{(1)}\right)$ we have,
\begin{equation}
    u\De^{(1)}
    \colon
    \left(
        \begin{gathered}
            \alpha_{0},\alpha_{1},\alpha_{2},\alpha_{3},
            \\
            \alpha_{4},\alpha_{5},\alpha_{6},\alpha_{7}
        \end{gathered}
    \right)
    \mapsto
    \left( 
        \begin{gathered}
            \alpha_{0},  \al_{1345},-\al_{2344556},\al_{6}, 
            \\
           \al_{245},\al_{3},\al_4,\al_{567},   
        \end{gathered}
\right).
    \label{uD}
\end{equation}
That is, we have 
\begin{equation}\label{uJd}
    uJ=u\{\al_0, \al_3, \al_5,\al_7\}
    =\{uJ_1, uJ_2, uJ_3, uJ_4\}=\{\al_0, \al_6, \al_3,\al_{567}\},
\end{equation}
which is another type $4A_1$ sub-root system of the $E_7^{(1)}$ root system $\Phi^{(1)}$.

\begin{proposition}\label{NJWJe72}
  The normalizer of ${W}_{uJ}$ with $uJ=\{\al_0, \al_6, \al_3,\al_{567}\}
  \subset \Phi^{(1)}$ in $\widetilde{W}( E_{7}^{(1)} )$ is given by
\begin{align}\label{NWJe72}
  N(W_{uJ})&=W_{uJ}\rtimes N_{uJ}=\lan s_0, s_6, s_3, s_{76567}\ran\rtimes\left(\lan ub_1, ub_2, ub_3, ub_0\ran\rtimes
  \lan u\sigma u^{-1}\ran\right)\\\nonumber
&\cong {W}(4A_1)\rtimes \widetilde{W}\left(B_3^{(1)}\right),  
\end{align}
with,
\begin{equation}\label{B3sr2}
    ub_1 = us_{6576}u^{-1},\quad
    ub_2 = us_{4534}u^{-1}, \quad
    ub_3 = us_{2}u^{-1}, \quad
    ub_0 = us_{1301}u^{-1},
\end{equation}
where $u=s_{5436542}$.
\end{proposition}
\begin{proof}
By Equations \eqref{sr} and \eqref{uJd},
generators of the parabolic subgroup $W_{uJ}$
are obtain from generators of $W_{J}$
by a conjugation of $u=s_{5436542}$, that is,
\begin{equation}
 W_{uJ}=uW_Ju^{-1}
 =\lan us_0u^{-1}, us_3u^{-1}, us_5u^{-1}, us_7u^{-1}\ran=\lan s_0, s_6, s_3, s_{76567}\ran.
 \end{equation}
By \cite{BH}, the normalizer of ${W}_{uJ}$
is obtained from the normalizer of ${W}_{J}$ by a conjugation of $u$.
\end{proof}
\begin{proposition}\label{NJWJe7r2}
    The root system for $N(W_{uJ})$ of \Cref{NJWJe72} is given by 
    \begin{align}\label{nr}
uJ\cup u\be=\{uJ_1, uJ_2, uJ_3, uJ_4\}\cup
\{u\be_1, u\be_2,u\be_3,u\be_0\}\cong 4A_1\times B_3^{(1)},
\end{align}
with
\begin{equation}\label{nrre7}
    \begin{gathered}
        uJ_1=\al_{0},\quad
        uJ_2=\al_{6},
        \\
        uJ_3=\al_{3},\quad
        uJ_4=\al_{567},
        \\
        u\be_1=\al_{344567},\quad
        u\be_2=\al_{22344556},
        \\
        u\be_3=-\al_{2344556},\quad
        u\be_0=\al_{1133445560},
    \end{gathered}
\end{equation}
where $u\be_0+u\be_1+2u\be_2+2u\be_3=\alpha_{071122663335554444}=\de$.
The group $N_{uJ}$ can be realised as 
    $\widetilde{W}(B_3^{(1)})$ with $ub_i$ acting as 
    reflection along the root $u\be_i$, while $u\sigma u^{-1}$ exchanges
    $u\be_0$ and $u\be_1$, see Figure \ref{rsB3ae}.
\end{proposition}

\begin{proposition}
    The map $\varphi_{2}^{*}$ is an element of quasi-translation of 
    order three in the
    \begin{equation}
        N(W_{uJ})\cong{W}(4A_1)\rtimes \widetilde{W}(B_3^{(1)})
        \label{NWuJ}
    \end{equation}
    subgroup of $\widetilde{W}(E_7^{(1)})$ give in \Cref{NJWJe72}. We have,
    \begin{equation}\label{phi2u}
        \varphi_{2}^{*}=u\sigma u^{-1}s_{0}s_{3}ub_{3231203}.    
    \end{equation}   
    \label{prop:phi2norm}
\end{proposition}

\begin{proof}
The actions of $\varphi_{2}^{*}$ on the $uJ\cup u\be$ basis
is obtained from that of \Cref{eq:phi3E71B}
by a change of basis,
\begin{equation}
    \varphi_{2}^{*}
    \colon
    \left(
        \begin{gathered}
            uJ_1, uJ_2, uJ_3, uJ_4,
            \\
            u\be_1, u\be_2,u\be_3,u\be_0
        \end{gathered}
    \right)
    \mapsto
    \left( 
        \begin{gathered}
            uJ_1, -uJ_4, -uJ_2, uJ_3, 
            \\
           u\be_2, u\be_{1222333300},-u\be_{230},-u\be_{122233330}  
        \end{gathered}
\right).
    \label{eq:phi2u}
\end{equation}
Moreover, we have
\begin{equation}
    \left(\varphi_{2}^{*}\right)^3
    \colon
    \left(
        \begin{gathered}
            uJ_1, uJ_2, uJ_3, uJ_4,
            \\
            u\be_1, u\be_2,u\be_3,u\be_0
        \end{gathered}
    \right)
    \mapsto
    \left( 
        \begin{gathered}
            uJ_1, uJ_2, uJ_3, uJ_4, 
            \\
           u\be_1, u\be_2,u\be_3+\de,u\be_0-2\de  
        \end{gathered}
\right).
    \label{eq:phi23u}
\end{equation}
We use Equation \eqref{lfun} on $\varphi_{2}^{*}$'s action on the $uJ\cup u\be$-system given in \Cref{eq:phi2u}
to write $\varphi_{2}^{*}$ as a product of generators of  $N(W_{uJ})$.
We find,
\begin{equation}
    \varphi_{2}^{*}=u\sigma u^{-1}s_{0}s_{3}ub_{3231203},
\end{equation}
thus ending the proof.
\end{proof}

Hence, the proof of last statement of \Cref{thm:mainB} follows
from \Cref{prop:phi2norm}.

\section{Conclusion}
\label{sec:conclusions}

In this paper, we investigated the symmetries of three plane integrable
systems given by \Cref{eq:qrt0,eq:phi3,eq:phis}. The first two systems,
$\varphi_{1}$~\eqref{eq:qrt0} and $\varphi_{2}$~\eqref{eq:phi3}
arise from the KHK discretisation of the tetrahedral and
octahedral reduced Nahm systems~\cite{Hitchinetal1995}. Such
systems and some generalisations have been discussed
extensively in the literature~\cite{PetreraPfadlerSuris2011,
CelledoniMcLachlanMcLarenOwrenQuispel2017, CarsteaTakenawa2013, GubNahm,
GJ_biquadratic, Zander2020} (see \Cref{app:equivalence}), but a close
inspection of their action on singular fibres was never discussed. The
last map, $\varphi_{s}$~\eqref{eq:phis}, is a generalisation of the
map $\varphi_{2}$ we introduced in \Cref{sec:origin} through the
QRT root construction (see~\cite[\S 6.3]{HietarintaJoshiNijhoff2016}
and~\cite[Introduction]{Duistermaat2011book}).

In this paper, in light of the theory of Oguiso and Shioda~\cite{OS1991}
and its application to plane discrete integrable systems,
see~\cite{Carsteaetal2017} and~\cite[Chap. 7]{Duistermaat2011book}, we
found out that these three systems possess singular fibre structure of
type ${(A_3\times A_2 \times 2A_1)}^{(1)}$, ${(3A_3\times A_1 )}^{(1)}$,
and ${(2A_{1}\times 2 A_{2} \times A_{0}^{*})}^{(1)}$ respectively.

After this construction, we concentrated our study on three specific
singular fibres and their symmetry group, i.e.\  their orthogonal
component in the Picard group, namely the fibre $[0:1]$ of type
$A_2^{(1)}$ for $\varphi_1$, the fibre $[0:1]$ of type $A_1^{(1)}$
for the map $\varphi_2$, and the fibre $[0:1]$ of type $A_1^{(1)}$
for the map $\varphi_s$. The result is that on these singular fibres
the symmetry group is a proper subgroup of symmetry group, forming a
sublattice in the Picard group. This situation is not completely unusual,
for instance see the papers~\cite{T:03, KNT:11, CarsteaTakenawa2012,
Carsteaetal2017, ahjn:16, cp:17, Stokes:18}. However, our result is
relevant because for the first time we find the non-simply laced Coxeter
group $B_3^{(1)}\subset E_7^{(1)}={(A_1^{(1)})}^\perp$ as symmetry
group for both the maps $\varphi_2$ and $\varphi_{s}$. Moreover, our
result is obtained within the algorithmic framework of the normaliser
of parabolic subgroups~\cite{H,BH}, which already proved fruitful in
interpreting known results~\cite{Shi:19,Shi:22}. In this sense, we can
say that non-simply laced symmetry groups were actually \emph{hidden in
plain sight} in known integrable plane systems.

Further research is ongoing about the symmetry types associated to the
other singular fibres. Furthermore, we note that it is also interesting
to consider the (possible) subgroups of elements of the generic
fibres, i.e.\ the elliptic fibre of type $A_{0}^{(1)}$. Considering
the deautonomisation of these systems we could then possibly build new
``degenerate'' elliptic discrete Painlev\`e equations as it was done
for instance in~\cite{Carsteaetal2017, ahjn:16, JoshiNakazono2017}. In
general, for each singular fibre the deautonomisation procedure
of~\cite{Carsteaetal2017} might be applied.

To conclude, what we laid down in this paper can also evolve into an
ambitious program of building the symmetry group of all the QRT maps
and QRT root associated with given elements of the classification given
by Oguiso and Shioda's classification of singular fibre configurations
of rational elliptic surfaces~\cite{OS1991}, following what was done
in~\cite[Chap. 7]{Duistermaat2011book}.

\subsection*{Acknowledgments}

GG acknowledges support of the GNFM through Progetto Giovani GNFM 2023:
``Strutture variazionali e applicazioni delle equazioni alle differenze
ordinarie'', CUP\_E53C22001930001.

We thank Prof. Bert Van Geemen and Dr. Michele Graffeo for their
assistance and stimulating discussion during the preparation of this
paper.

\appendix

\section{Reduced Nahm systems of tetrahedral and octahedral type}
\label{app:equivalence}

In \cite{Hitchinetal1995}, some special cases of Nahm's equations
with particular symmetries studied in connection with
the theory of monopoles. Due to the symmetry of the associated Nahm 
matrices the systems are called the \emph{tetrahedral Nahm system}, 
\emph{octahedral Nahm system}, and \emph{icosahedral Nahm system}.
In particular, we adopt the form of the systems of \emph{tetrahedral Nahm system} 
and \emph{octahedral Nahm system} from \cite{PetreraPfadlerSuris2011}:
\begin{subequations}
    \begin{gather}
        \dot{x}= x^2-y^2, 
        \quad 
        \dot{y} = -2 x y,
        \label{eq:tetra}
        \\
        \dot{x} = 2 x^2-48 y^2,
        \quad
        \dot{y}= -6 x y-8y^2.
        \label{eq:octa}
    \end{gather}\label{eq:rnahm}
\end{subequations}

Consider now the family of Hamiltonian systems:
\begin{equation}
    H_{a,b,c} = 
    x^{a}y^{b}\left( k x + l y \right)^{c},
    \label{eq:Habc}
\end{equation}
see for instance~\cite{CelledoniMcLachlanMcLarenOwrenQuispel2017}.
The corresponding Hamilton equations are given by:
\begin{equation}
    \begin{pmatrix}
        \dot{x}
        \\
        \dot{y}
    \end{pmatrix}
    = 
    x^{1-a}
    y^{1-b}
    (kx+ly)^{1-c}
    \begin{pmatrix}
        0 & 1
        \\
        -1 & 0
    \end{pmatrix} \grad H_{a,b,c}=
    \begin{pmatrix}
        x (b k x+b l y+c l y), 
        \\
        -y (a k x+a l y+c k x)
    \end{pmatrix}.
    \label{eq:eqHabc}
\end{equation}
This system is \emph{quadratic} and \emph{integrable} for all
triples $(a,b,c)$ (it preserves the invariant \eqref{eq:Habc}).
By direct computation we have the following result:

\begin{lemma}
    The Hamiltonian systems \eqref{eq:eqHabc} with $(a,b,c)=(1,1,1)$
    is mapped to the tetrahedral Nahm system \eqref{eq:tetra} through
    the following transformation:
    \begin{equation}
        (x,y)\to \left(\frac{\sqrt{3}(4kl)^{1/3}}{2k}x
        -\frac{(4kl)^{1/3}}{2k}y, \frac{(4kl)^{1/3}}{l}y\right),
        \quad
        t\to \frac{\sqrt(3)(4kl)^{1/3}}{2} t,
    \end{equation}
    while the Hamiltonian systems \eqref{eq:eqHabc} with 
    $(a,b,c)=(1,1,2)$ is mapped to the octahedral Nahm systems \eqref{eq:octa} 
    through the following transformation:
    \begin{subequations}
        \begin{gather}
        (x,y)\to
        \left(
            -\frac{1+\imath}{5} \left(\frac{5}{k}\right)^{3/4} l^{1/4}x
            -\frac{3}{10}(1+\imath) l^{1/4} \left(\frac{5}{k}\right)^{3/4}y, 
            \frac{1+\imath}{2}\left(\frac{5}{l}\right)^{3/4}k^{1/4} y
        \right)
        \\
        t \to -\frac{1+\imath}{10} 5^{3/4} (lk)^{1/4} t
        \end{gather}
    \end{subequations}
\end{lemma}

Let us now turn to the genesis of the maps $\varphi_{1}$ and
$\varphi_{2}$. As stated in the introduction the KHK 
discretisation~\cite{Kahan1993,KahanLi1997,HirotaKimura2000,KimuraHirota2000} 
is a discretisation method for systems of first-order differential
equations.

\begin{definition}
    Consider a system of first-order ordinary differential equations:
    \begin{equation}
        \vec{\dot{x}}=\vec{f}\left( \vec{x} \right),
        \quad
        \vec{x}\colon\R^{N}\to\R,
        \quad
        \vec{f}\colon\R^{N}\to\R^{N},
        \quad
        N\in\N,
        \,
        t\in\R.
        \label{eq:firstord}
    \end{equation}
    Then the \emph{KHK discretisation of the system 
    \eqref{eq:firstord}} is given by:
    \begin{equation}
        \frac{\vec{x'}-\vec{x}}{h} = 
        2\vec{f}\left( \frac{\vec{x'}+\vec{x}}{2} \right)
        -\frac{\vec{f}\left( \vec{x'} \right)+\vec{f}\left( \vec{x}\right)}{2},
        \label{eq:kahan}
    \end{equation}
    where $\vec{x} = \vec{x}\left( n h \right)$,
    $\vec{x'}=\vec{x}\left( (n+1)h \right)$, and $h$ is a infinitesimal
    parameter, i.e. $h\to0^{+}$.
    \label{def:khk}
\end{definition}

If the function $\vec{f}$ is \emph{quadratic}, then the discrete
system defined by \Cref{eq:kahan} give rise to a birational map of 
$\R^N$, see \cite{PetreraPfadlerSuris2011,CelledoniMcLachlanOwrenQuispel2013}.
In such a case, this birational map can be considered to be a 
restriction of a map over $\Cp^N$, but dynamical properties are better
understood over compact spaces to take into account points escaping
at infinity. So, it is costumary to consider the maps over some 
compactification of $\Cp^N$. The usual choices are $\Pj^{M}$ or 
$(\Pj^{1})^{\times M}$. In this paper, we will use the latter, but 
each choice will yield the same final results.

Now following~\cite{VanDerKampCelledoniMcLachlanMcLarenOwrenQuispel2019}\
we show that the KHK maps corresponding to the discretisation cases
$(a,b,c)=(1,1,2)$ and $(a,b,c)=(1,1,1)$ of the systems~\eqref{eq:eqHabc}\ 
can be brought to the form~\eqref{eq:qrt0}\ and~\eqref{eq:phi3}. That is, 
denoting these two maps by $\bPhi_{(1,1,1),h}\in\Bir(\PcrossP)$ and
$\bPhi_{(1,1,2),h}\in\Bir(\PcrossP)$ respectively, we have by a direct 
computation:

\begin{proposition}
    The maps $\bPhi_{(1,1,2),h}\in\Bir(\PcrossP)$:
    \begin{equation}
    \bPhi_{(1,1,2),h}\colon
    (x,y)\mapsto
    \begin{pmatrix}
            \frac{x\left[2+\left(3 k x+5 l y\right)h\right]}{%
            2+\left( k x+ l y\right)-\left(3 k^2 x^2+2 k l x y+3 l^2 y^2\right)h^{2}}
            \\
             \frac{y \left[2-\left(5  k x+3  l y\right)\right]}{
           2+\left( k x+ l y\right)-\left(3 k^2 x^2+2 k l x y+3 l^2 y^2\right)h^{2}}
    \end{pmatrix}^{T}
    \label{eq:rnahm112}
    \end{equation}
    and
    $\bPhi_{(1,1,1),h}\in\Bir(\PcrossP)$:
    \begin{equation}
    \bPhi_{(1,1,1),h}\colon
    (x,y)\mapsto
    \begin{pmatrix}
            \frac{x[1+ (k x+2 l y)h]}{1- h^2 k^2 x^2- h^2 k l x y- h^2 l^2 y^2},
            \\
            \frac{y [1- (2  k x-  l y)h]}{1- h^2 k^2 x^2- h^2 k l x y- h^2 l^2 y^2}
    \end{pmatrix}^{T}
    \label{eq:rnahm111}
    \end{equation}
    are \emph{conjugated} to the maps
    $\varphi_{1}$ and $\varphi_{2}$ respectively, through a projective
    coallineation.
    That is:
    \begin{equation}
        \varphi_{1} = \chi_{1}^{-1} \circ \bPhi_{(1,1,2),h} \circ \chi_{1},
        \quad\text{and}\quad
        \varphi_{2} = \chi_{2}^{-1} \circ \bPhi_{(1,1,1),h} \circ \chi_{2}.
        \label{eq:conj}
    \end{equation}
    where:
    \begin{subequations}
        \begin{align}
        \chi_{1} &\colon
        \left( x,y \right)\mapsto
        \left( 
        \frac{1}{2kh} \frac{2 (x- y)+\sqrt{2}}{x+y}, 
        -\frac{1}{2lh} \frac{2 (x- y)-\sqrt{2}}{x+y}
        \right),
        \label{eq:pj}
        \\
        \chi_{2} &\colon
        \left( x,y \right)\mapsto
        \left( 
        \frac{2}{x+y}, -\frac{hk-x+y}{hl(x+y)}
        \right).
        \label{eq:pj3}    
        \end{align}
    \end{subequations}
    \label{eq:shape}
\end{proposition}

\begin{remark}
    We remark that the literature on the KHK discretisation of the Nahm
    systems~\eqref{eq:rnahm} is extensive. In this remark we collect some
    sparse facts that are worth mentioning:

    \begin{itemize}
        \item The integrability of the KHK discretisation of the three 
            kinds of reduced Nahm systems, tetrahedral, octahedral and 
            icosahedral, was proved by direct computation 
            in~\cite{PetreraPfadlerSuris2011}. A Lax pair of these systems
            was introduced in~\cite{GubNahm}.
        \item In~\cite{CarsteaTakenawa2013} it was noted that these 
            systems are not minimal, in the sense of \Cref{sec:backgrownd},
            and can be reduced to QRT roots.
        \item Generalisations of these systems in $N$-dimensions were 
            presented in~\cite{PetreraZander2017}.
        \item From~\cite{Zander2020} it is known that the only integrable 
            cases of KHK discretisation of the system~\eqref{eq:eqHabc} 
            are obtained if $(a,b,c)=(1,1,1)$, $(a,b,c)=(1,1,2)$, 
            $(a,b,c)=(1,2,3)$ and their permutations.
        \item The map $\chi_1$ given in~\eqref{eq:pj} is a modification
            of a map already presented in~\cite{GJ_biquadratic}.
    \end{itemize}
\end{remark}

\section{Basic properties of the group $\widetilde{W}( E_{6}^{(1)} )$}\label{app:E6}

The Dynkin diagram of type $E_6^{(1)}$, 
$\Ga(E_6^{(1)})$ is given in Figure \ref{fig:e6}.
Its simple system $\De^{(1)}=\{\al_j\mid 0\leq j\leq 6\}$ forms a basis for an $7$-dimensional real vector space
$V^{(1)}$, equipped with 
a bilinear form given
by Equation \eqref{alaij0} and
the generalized Cartan matrix of type $E_6^{(1)}$,
\begin{align}\nonumber
C(E_6^{(1)})&=(a_{ij})_{1\leq i,j\leq 6, 0}
=(\al_i\cdot\oc\al_j)_{1\leq i,j\leq 6, 0},\\\nonumber
&=(\al_i\cdot\al_j)_{1\leq i,j\leq 6, 0},\\\label{CarE6a}
&=\left(
\begin{array}{cccccc}
 2 & -1 & 0 & 0 & 0 & 0 \\
 -1 & 2 & -1 & 0 & 0 & 0 \\
 0 & -1 & 2 & -1 & 0 & -1 \\
 0 & 0 & -1 & 2 & -1 & 0 \\
 0 & 0 & 0 & -1 & 2 & 0 \\
 0 & 0 & -1 & 0 & 0 & 2 \\
\end{array}
\right).    
\end{align}
Defining relations of
$W(E_6^{(1)})=\lan s_i\mid 0\leq i \leq 6\ran$ 
can be read off $\Ga(E_6^{(1)})$. Generators $s_j\in W(E_6^{(1)})$ act on $V^{(1)}$ by Equation \eqref{sij0}, where $a_{ij}$ is the $(i,j)$-entry
of $C(E_6^{(1)})$ in Equation \eqref{CarE6a}. We have the null root and the highest root of $E_6^{(1)}$ given by
\begin{equation}\label{deE6}
    \de=\al_0+\tilde{\al}=\al_0++\alpha_{1}+\alpha_{5}
    +2 \left( \alpha_{2}+\alpha_{4}+\alpha_{6} \right)
    +3\alpha_{3}.
\end{equation}

A dual space $V^{(1)\ast}$ with basis $\{ h_1, \ldots, h_6, h_\de\}$ is given by \Cref{ah2}. The subspace $X_0\subset V^{(1)\ast}$ has two bases:
the set of simple coroots $\{\pi(\oc\al_j)\mid 1\leq j\leq 6\}$,
and the set of fundamental weights
$\{h_j\mid1\leq j\leq 6\}$.
They are related by Equation \eqref{pah}, 
\beq\label{pahE6}
\al_j
=\sum_{k=1}^{6}\left(C(E_6)\right)_{kj}h_k
=\sum_{k=1}^{6}a_{jk}h_k, \quad 1\leq j\leq 6,
\eeq
where we have used the fact that 
$\Ga(E_6^{(1)})$ is simply-laced, 
so $C(E_6)$ is symmetric and
we have identified the simple coroots $\pi(\oc\al_j)$ with the simple roots $\al_j$ for $1\leq j\leq 6$.
The squared lengths of $h_j (1\leq j\leq 6)$ is given by
the diagonal entries of $C(E_6)^{-1}$ in Equation \eqref{hijE6},
\beq\label{hijE6}
\left((h_i, h_j)\right)_{1\leq i,j\leq6}
=C(E_6)^{-1}=
\left(
\begin{array}{cccccc}
 \frac{4}{3} & \frac{5}{3} & 2 & \frac{4}{3} & \frac{2}{3} & 1 \\
 \frac{5}{3} & \frac{10}{3} & 4 & \frac{8}{3} & \frac{4}{3} & 2 \\
 2 & 4 & 6 & 4 & 2 & 3 \\
 \frac{4}{3} & \frac{8}{3} & 4 & \frac{10}{3} & \frac{5}{3} & 2 \\
 \frac{2}{3} & \frac{4}{3} & 2 & \frac{5}{3} & \frac{4}{3} & 1 \\
 1 & 2 & 3 & 2 & 1 & 2 \\
\end{array}
\right).
\eeq
That is they are $\frac{4}{3}, \frac{10}{3}, 6, \frac{10}{3}, \frac{4}{3}$
and $2$, respectively.
Moreover, by Equation \eqref{hpa} we have
\beq\label{hpaE6}
h_i=\sum_{k=1}^6C(E_6)^{-1}_{ik}\al_k\quad
\mbox{for}\quad1\leq i\leq n.
\eeq
with $C(E_6)^{-1}$ in Equation \eqref{hijE6}.
The group $A$ of diagram automorphsims of
$\Ga(E_6^{(1)})$ is generated by $\sigma$,
written
as a permutation on the index set of $\De^{(1)}$ we have,
\begin{equation}\label{4ddae6}
\sigma=(150)(246),
\end{equation}
so that $(\sigma)^3=1$. Moreover, we have
\begin{equation}\label{ddae6s}
    \sigma s_{\{1,2,3,4,5,6,0\}}=s_{\{5,4,3,6,0,2,1\}}\sigma.
\end{equation}
That is, $\sigma$ permutes the simple reflections by conjugation, for example,
we have $\sigma s_1 \sigma^{-1}=s_5$ from Equation \eqref{ddae6s}.
The extended affine Weyl group 
of type $E_{6}$ is given by 
\begin{equation}\label{eae6}
  \widetilde{W}(E_{6}^{(1)})=\langle s_i\;|\;0\leq i \leq 6\rangle\rtimes\langle \sigma\rangle=W(E_{6})\ltimes\lan u_j\, \mid 1\leq j\leq 6\ran,
\end{equation}
where $u_j$ is the element of translation by $h_j$.

\section{Basic properties of the group $\widetilde{W}( E_7^{(1)} )$}\label{app:E7}

The Dynkin diagram of type $E_7^{(1)}$, $\Ga(E_7^{(1)})$ is given
in Figure \ref{fig:e7B}.  Its simple system $\De^{(1)}=\{\al_j\mid
0\leq j\leq 7\}$ forms a basis for an $8$-dimensional real vector space
$V^{(1)}$ equipped with a bilinear form given by Equation \eqref{alaij0}
and the generalized Cartan matrix of type $E_7^{(1)}$:
\begin{align}\nonumber
C(E_7^{(1)})&=(a_{ij})_{1\leq i,j\leq 7, 0}
=(\al_i\cdot\oc\al_j)_{1\leq i,j\leq 7, 0},\\\nonumber
&=(\al_i\cdot\al_j)_{1\leq i,j\leq 7, 0},\\\label{CarE7a}
&=\left(
\begin{array}{cccccccc}
 2 & 0 & -1 & 0 & 0 & 0 & 0 & -1 \\
 0 & 2 & 0 & -1 & 0 & 0 & 0 & 0 \\
 -1 & 0 & 2 & -1 & 0 & 0 & 0 & 0 \\
 0 & -1 & -1 & 2 & -1 & 0 & 0 & 0 \\
 0 & 0 & 0 & -1 & 2 & -1 & 0 & 0 \\
 0 & 0 & 0 & 0 & -1 & 2 & -1 & 0 \\
 0 & 0 & 0 & 0 & 0 & -1 & 2 & 0 \\
 -1 & 0 & 0 & 0 & 0 & 0 & 0 & 2 \\
\end{array}
\right).    
\end{align}
Defining relations for
$W(E_7^{(1)})=\lan s_i\mid 0\leq i \leq 7\ran$ 
can be read off $\Ga(E_7^{(1)})$. Generators $s_j\in W(E_7^{(1)})$ act on $V^{(1)}$ by Equation \eqref{sij0}, where $a_{ij}$ is the $(i,j)$-entry
of $C(E_7^{(1)})$ from Equation \eqref{CarE7a}. We have the null root and the highest root of $E_7^{(1)}$ given by
\begin{equation}\label{deE7}
    \de=\al_0+\tilde{\al}=\al_0++\alpha_{7}
    +2 \left( \alpha_{2}+\alpha_{1}+\alpha_{6} \right)
    +3 \left(\alpha_{3}+ \al_5\right)+4\al_4.
\end{equation}

A dual space $V^{(1)\ast}$ with basis $\{ h_1, \ldots, h_7, h_\de\}$ is given by \Cref{ah2}. The subspace $X_0\subset V^{(1)\ast}$ has two bases:
the set of simple coroots $\{\pi(\oc\al_j)\mid 1\leq j\leq 7\}$,
and the set of fundamental weights
$\{h_j\mid1\leq j\leq 7\}$, and
they are related by Equation \eqref{pah}, 
\beq\label{pahE7}
\al_j
=\sum_{k=1}^{7}\left(C(E_7)\right)_{kj}h_k
=\sum_{k=1}^{7}a_{kj}h_k, \quad 1\leq j\leq 7,
\eeq
where we have used the fact that 
$\Ga(E_7^{(1)})$ is simply-laced, 
so $C(E_7)$ is symmetric and
we have identified the simple coroots $\pi(\oc\al_j)$ with the simple roots $\al_j$ for $1\leq j\leq 7$.
The squared lengths of $\{h_j\mid1\leq j\leq 7\}$ is given by
the diagonal entries of $C(E_7)^{-1}$ in Equation \eqref{hijE7},
\beq\label{hijE7}
\left((h_i, h_j)\right)_{1\leq i,j\leq7}
=C(E_7)^{-1}=
\left(
\begin{array}{ccccccc}
 2 & 2 & 3 & 4 & 3 & 2 & 1 \\
 2 & \frac{7}{2} & 4 & 6 & \frac{9}{2} & 3 & \frac{3}{2} \\
 3 & 4 & 6 & 8 & 6 & 4 & 2 \\
 4 & 6 & 8 & 12 & 9 & 6 & 3 \\
 3 & \frac{9}{2} & 6 & 9 & \frac{15}{2} & 5 & \frac{5}{2} \\
 2 & 3 & 4 & 6 & 5 & 4 & 2 \\
 1 & \frac{3}{2} & 2 & 3 & \frac{5}{2} & 2 & \frac{3}{2} \\
\end{array}
\right).
\eeq
That is, they are $2$, $\frac{7}{2}$, $6$, $12$, $\frac{15}{2}$, $4$, $\frac{3}{2}$.
Moreover, by Equation \eqref{hpa} we have
\beq\label{hpaE7}
h_i=\sum_{k=1}^7C(E_7)^{-1}_{ik}\al_k\quad
\mbox{for}\quad1\leq i\leq n,
\eeq
with $C(E_7)^{-1}$ given in Equation \eqref{hijE7}.
The group $A$ of diagram automorphsims of
$\Ga(E_7^{(1)})$ is generated by $\sigma$,
written
as a permutation on the index set of $\De^{(1)}$ we have:
\begin{equation}\label{ddae7}
\sigma=(07)(16)(35)
\end{equation}
so that $(\sigma)^2=1$. Moreover, we have
\begin{equation}\label{ddae7s}
    \sigma s_{\{1,2,3,4,5,7,0\}}=s_{\{6,2,5,4,3,1,0\}}\sigma.
\end{equation}
The extended affine Weyl group 
of type $E_{7}$ is then 
\begin{equation}\label{eae7}
  \widetilde{W}(E_{7}^{(1)})=\langle s_i\;|\;0\leq i \leq 7\rangle\rtimes\langle \sigma\rangle=W(E_{7})\ltimes\lan u_j\, \mid 1\leq j\leq 7\ran,
\eeq
where $u_j$ is the element of translation by $h_j$.

\printbibliography

\end{document}